\documentclass[12pt, draftclsnofoot, onecolumn]{IEEEtran}
\ifCLASSINFOpdf

\else

\fi
\usepackage{amsmath}
\usepackage{makeidx}  
\usepackage{algorithm}
\usepackage{algorithmic}
\usepackage{graphicx}
\usepackage{subfigure}
\usepackage{epstopdf}
\usepackage{bm}
\usepackage{bbding}
\usepackage{cite}
\usepackage{stfloats}

\usepackage{amssymb}
\usepackage{graphicx}

\usepackage{url}
\newtheorem{proposition}{Proposition}
\usepackage{tabularx,booktabs}
\newcolumntype{C}{>{\centering\arraybackslash}X} 
\usepackage{lipsum}

\usepackage{makecell} 

\usepackage{graphicx}
\usepackage{color}
\newtheorem{thm}{Theorem}
\newtheorem{rem}{Remark}
\newtheorem{lem}{Lemma}
\newtheorem{proof}{proof}

\begin{document}

\title{IRS-aided Wireless Powered MEC Systems: TDMA or NOMA for Computation Offloading?}

\author{}

\author{Guangji~Chen,~\IEEEmembership{}
        Qingqing~Wu,~\IEEEmembership{Member,~IEEE,}
       Wen~Chen,~\IEEEmembership{Senior Member,~IEEE,}
       Derrick~Wing~Kwan~Ng,~\IEEEmembership{Fellow,~IEEE,}
       and Lajos~Hanzo,~\IEEEmembership{Fellow,~IEEE}
        \thanks{G. Chen and Q. Wu are with the State Key Laboratory of Internet of Things for Smart City, University of Macau, Macao 999078, China (email: guangjichen@um.edu.mo; qingqingwu@um.edu.mo).

         W. Chen is with the Department of Electronic Engineering,
         Shanghai Institute of Advanced Communications and Data Sciences, Shanghai
         Jiao Tong University, Minhang 200240, China (e-mail: wenchen@sjtu.edu.cn).

         Derrick Wing Kwan Ng is with the School of Electrical Engineering and
         Telecommunications, UNSW Sydney, Sydney, NSW 2052, Australia (e-mail:
         w.k.ng@unsw.edu.au).

         L. Hanzo is with the Department of Electronics and Computer Science,
         University of Southampton, Southampton SO17 1BJ, U.K. (e-mail:
         lh@ecs.soton.ac.uk).}
         }


\maketitle
\vspace{-25pt}
\begin{abstract}
An intelligent reflecting surface (IRS)-aided wireless powered mobile edge computing (WP-MEC) system is conceived, where each device's computational task can be divided into two parts for local computing and offloading to mobile edge computing (MEC) servers, respectively. Both time division multiple access (TDMA) and non-orthogonal multiple access (NOMA) schemes are considered for uplink (UL) offloading. Given the capability of IRSs in intelligently reconfiguring wireless channels over time, it is fundamentally unknown which multiple access scheme is superior for MEC UL offloading. To answer this question, we first investigate the impact of three different dynamic IRS beamforming (DIBF) schemes on the computation rate of both offloading schemes, based on the flexibility for the IRS in adjusting its beamforming (BF) vector in each transmission frame. Under the DIBF framework, computation rate maximization problems are formulated for both the NOMA and TDMA schemes, respectively, by jointly optimizing the IRS passive BF and the resource allocation. We rigorously prove that offloading adopting TDMA can achieve the same computation rate as that of NOMA, when all the devices share the same IRS BF vector during the UL offloading. By contrast, offloading exploiting TDMA outperforms NOMA, when the IRS BF vector can be flexibly adapted for UL offloading. Despite the non-convexity of the computation rate maximization problems for each DIBF scheme associated with highly coupled optimization variables, we conceive computationally efficient algorithms by invoking alternating optimization. Our numerical results demonstrate the significant performance gains achieved by the proposed designs over various benchmark schemes and also unveil that the optimal time allocated to downlink wireless power transfer can be effectively reduced with the aid of IRSs, which is beneficial for both the system's spectral efficiency and its energy efficiency.
\end{abstract}



\IEEEpeerreviewmaketitle
\section{Introduction}
With the rapid development of popular Internet-of-Everything (IoE) technologies, the unprecedented proliferation of mobile sensors, electronic tablets, and wearable devices is set to continue in support of smart transportation, smart homes, and smart cities \cite{8879484}. For realizing the IoE, next generation wireless networks are expected to support massive number of connections and accommodate huge data traffic. As such, superior multiple access (MA) schemes are required to attain high spectral efficiency (SE) for a massive number of IoE devices in next generation wireless networks\cite{8114722}. Recently, it has been shown that non-orthogonal multiple access (NOMA) is capable of improving the SE by allowing multiple users to simultaneously access the same spectrum. Therefore, NOMA has been recognized as one of the key technologies in next generation wireless networks\cite{8114722}.

On the other hand, an IoE device is often equipped with a low-performance processor and limited battery capacity, given their practical size and cost constraints. In particular, the emerging applications, such as unmanned driving and automatic navigation, generally rely on the execution of low-latency and computation-intensive tasks, thus imposing new challenges on IoE devices \cite{8016573}. As a remedy, by incorporating radio frequency (RF)-based wireless power transmission (WPT) and mobile edge computing (MEC), wireless powered MEC (WP-MEC) becomes a promising solution for granting self-sustainability and high computational capabilities to IoE systems \cite{7442079, 8304010, 8334188, 8434285, 9140412, 8986845, 9312671, 8537962, 9179779,8234686}. Specifically, RF-based WPT enables energy harvesting (EH) from RF signals and it is capable of prolonging the battery recharge-period of devices \cite{8234686, 7843670, 8421584}. To improve the computational capabilities for IoE systems, the MEC technology enables IoE devices to offload their tasks to nearby MEC servers in real time, which can compute their tasks remotely \cite{8016573}.

To enhance the computational efficiency of traditional WP-MEC systems, sophisticated resource allocation relying on optimization objectives, such as computation rate maximization \cite{8304010, 8334188, 8434285}, energy consumption \cite{8986845, 9312671, 8537962}, and latency minimization \cite{9179779}, etc, has been proposed. For instance, in \cite{7442079}, the WP-MEC framework was proposed for a single-user setup, where the probability of successfully processing a given amount of data was maximized subject to both end-to-end latency and EH constraints. In general, MEC supports a pair of basic operational modes, namely binary and partial offloading modes \cite{7442079, 8304010, 8334188, 8434285, 9140412}. Specifically, for the partial offloading mode, the computational task can be divided into two parts for partial local computing and offloading, respectively, while for the binary offloading mode, the computational task cannot be partitioned, hence it is either executed at the local device or offloaded to MEC servers \cite{9140412}. Based on the concept of  binary and partial offloading modes, the corresponding computation rate maximization problem was investigated in \cite{8304010, 8334188, 8434285 } for a multi-user setup, where time division multiple access (TDMA) was adopted for uplink (UL) offloading. As a further advance, the superiority of employing NOMA over TDMA in WP-MEC systems was quantified in terms of its energy efficiency improvement \cite{8986845, 9312671, 8537962} and latency reduction \cite{9179779}. Therefore, NOMA is regarded as an attractive scheme for UL offloading in traditional WP-MEC systems.

However, the efficiency of both the downlink (DL) WPT and UL offloading may become severely degraded by the wireless channel attenuation between transceivers, which thus fundamentally limits the performance of WP-MEC systems. With the goal of tacking this issue, the authors of \cite{8234686} exploited the multiple-input multiple-output (MIMO) technique for improving the WPT efficiency and studied the corresponding energy consumption minimization problem. Although the massive MIMO technology considerably improves the efficiency of both WPT and offloading by exploiting the huge beamforming (BF) gain \cite{8271992}, \cite{9027954}, the associated high hardware cost and energy consumption are still grave obstacles in the way of its practical implementation. Recently, intelligent reflecting surfaces (IRSs) have been proposed as a cost-effective technology for improving the spectral efficiency and energy efficiency of next generation wireless networks \cite{8910627, 9326394, 9424177}. Specifically, an IRS is a planar array comprised of a large number of low-cost passive reflecting elements, which can reflect incident signals and intelligently adapt their phase shifts according to the real-time channel conditions\cite{8910627}. As such, IRSs are capable of reconfiguring the wireless propagation environment for achieving e.g., signal enhancement and/or interference suppression. In particular, the fundamental squared-power gain of IRSs was originally unveiled in \cite{8811733}, which then inspired intensive research interests in investigating various IRS-aided wireless systems.

The new research paradigms of IRS-aided wireless information transmission (WIT), WPT, and MEC have been extensively studied in the literature. For IRS-aided WIT systems, joint passive BF at IRSs and active BF at base stations (BSs) was designed either for minimizing the transmit power of BSs or for maximizing the system capacity, e.g., \cite{8811733, 9139273, 9110869, 9279253, 9039554, 9427474}. As a further practical development, the analysis and optimization of IRS-aided wireless communications were studied by considering both discrete phase shifts \cite{8930608, 9295369} and amplitude-dependent phase shifts \cite{9115725}. In addition to exploiting IRSs for improving the WIT performance, the IRS technology is also appealing for WPT in IoE applications, thanks to its beneficial passive BF gain. Specifically, a promising line of research focused on passive BF design for simultaneous wireless information and power transfer (SWIPT) systems \cite{8941080, 9257429, 9133435}. Another line of research investigated IRS-aided wireless powered communication networks (WPCNs), where the devices first harvest energy in the DL and then transmit information in the UL \cite{9214497, 9298890, 9400380}. However, in traditional MEC systems, the task offloading efficiency may not be satisfactory due to the harsh propagation conditions of the wireless links. To address this issue, the authors of \cite{9270605,9133107} exploited the IRS technology for improving the offloading efficiency of MEC systems by studying the associated computation rate maximization and offloading latency minimization problems, respectively.

Given the aforementioned benefits of the IRS technique, its employment in WP-MEC systems is attractive for realizing IoE, since both the efficiency of DL WPT and UL offloading can potentially be improved. Additionally, next generation wireless networks require superior MA schemes for supporting a massive number of IoE devices. Therefore, integrating IRSs with efficient MA schemes in WP-MEC systems is essential for granting satisfactory experience of IoE applications in next generation wireless networks. For traditional MEC and WP-MEC systems operating without IRSs, the authors of \cite{8986845, 9312671, 8537962} demonstrated the superiority of NOMA over TDMA for UL offloading under the assumption of given wireless channels within a channel coherence time duration. However, these results may not be applicable to the new family of IRS-aided WP-MEC systems, since the IRS is capable of proactively establishing favorable time-varying wireless channels, which could introduce different impacts on MA schemes. As such, it still remains unknown which MA scheme is more efficient for UL offloading in IRS-aided WP-MEC systems. This knowledge-gap motivates us to investigate the achievable computation rate in such scenarios by considering the interplay between IRSs and MA schemes. To characterize the achievable computation rate of IRS-aided WP-MEC systems, the main challenges we identify are as follows: 1) the performance comparison between NOMA and TDMA  has to be carried out for IRS-aided UL offloading by considering the favorable time-varying wireless channels controlled by IRSs; 2) the specific IRS configuration required for reaping the potential benefits of WP-MEC systems has to be identified; 3) the design of IRS BF and resource allocation for WP-MEC systems is generally intractable.

To address the above issues, this paper investigates the achievable computation rate maximization problems of IRS-aided WP-MEC systems by considering two types of offloading schemes, i.e., TDMA and NOMA. Specifically, we focus our attention on a typical setup, where a hybrid access point (HAP) is exploited both as the energy transmitter and the MEC server. Moreover, an IRS is deployed for enhancing the efficiency of both DL WPT and UL offloading.
Our main contributions are summarized as follows:
\begin{itemize}
  \item We propose an offloading framework for investigating the performance of IRS-aided WP-MEC systems, where three different levels of dynamic IRS beamforming (DIBF) schemes are considered: \textbf{Case 1}: both DL WPT and UL offloading share the same IRS BF vector; \textbf{Case 2}: two different IRS BF vectors are exploited for DL WPT and UL offloading, respectively; \textbf{Case 3}: the IRS BF vectors can be further adapted for UL offloading with respect to each individual device. Under this framework, we formulate the corresponding computation rate maximization problems by jointly optimizing the resource allocation and the IRS BF for the aforementioned three cases.
  \item We analytically show that appropriately adjusting the IRS BF vectors for UL offloading is capable of improving the computation rate of TDMA, while it is not beneficial for that of NOMA. By analyzing the relationship between the computation rate maximization problems of TDMA and NOMA, we prove that the computation rate achieved by TDMA is the same as that by NOMA for both \textbf{Case 1} and \textbf{Case 2}. By contrast, since TDMA-based IRS-aided WP-MEC systems are capable of benefiting from varying the IRS BF vectors in the UL offloading stage, the computation rate of TDMA exceeds that of NOMA for \textbf{Case 3}.
  \item To gain insights into the beneficial effect of IRSs on WP-MEC systems, we first consider a single-user setup, where we derive a threshold-based UL offloading activation condition. Specifically, we demonstrate that UL offloading is activated iff the transmit power of the HAP is above a certain threshold and increasing the number of IRS elements is capable of reducing the threshold. For the more general multi-user setup, we develop an efficient alternating optimization (AO) algorithm for solving the resultant problems, where the resource allocation and the IRS BF design subproblems are solved alternatingly.
  \item Our numerical results show that the proposed IRS-aided WP-MEC designs are capable of substantially improving the computation rate compared to the benchmark schemes. It is also found that exploiting IRSs not only increases the total energy harvested via DL WPT, but also leaves more time available for UL offloading, which unveils a further benefit of IRSs for WP-MEC systems. Moreover, the computation rate of \textbf{Case 3} significantly exceeds that of \textbf{Case 1} and \textbf{Case 2}, while the performance loss of \textbf{Case 1} compared to \textbf{Case 2} is negligible. The results imply that the associated signaling overhead can be reduced by opting for \textbf{Case 1} instead of \textbf{Case 2} at the cost of a modest performance erosion.
\end{itemize}

The rest of this paper is organized as follows. Section II presents our system model and problem formulations. Section III provides the theoretical performance comparison of NOMA and TDMA-based UL offloading. Section IV focuses on studying the impact of IRSs on the UL offloading activation condition. Section V proposes computationally efficient algorithms for solving the formulated problems for the different scenarios. Section VI provides numerical results for evaluating the proposed designs. Finally, Section VII concludes the paper.

\emph{Notations:} $\mathbb{C}^{x \times y}$ stands for the set of complex $x \times y$ matrix. $\mathbb{Z}^ +$ represents the set of positive number. For a complex-valued vector ${\bf{a}}$, the $n$-th entry is denoted by ${\left[ {\bf{a}} \right]_n}$, ${{\bf{a}}^H}$ and ${{\bf{a}}^T}$ denote Hermitian transpose and transpose, respectively, ${\mathop{\rm diag}\nolimits} \left( {\bf{a}} \right)$ denotes a diagonal matrix with each diagonal entry being the corresponding entry in ${\bf{a}}$. The real part and the phase of a complex number $c$ are denoted by ${\mathop{\rm Re}\nolimits} \left( c \right)$ and $\arg \left( c \right)$, respectively. $\mathcal{O}\left(  \cdot  \right)$ is the big-O computational complexity notation. ${{\partial f\left(  \cdot  \right)} \mathord{\left/
 {\vphantom {{\partial f\left(  \cdot  \right)} {\partial x}}} \right.
 \kern-\nulldelimiterspace} {\partial x}}$ denotes the partial derivation of a function ${f\left(  \cdot  \right)}$ with respect to a variable $x$.


\section{System Models and Problem Formulations}

\subsection{System Model}
\vspace{-4pt}
As shown in Fig. \ref{model}, an IRS-aided WP-MEC system is considered, which consists of a HAP, an IRS, and $K$ wireless-powered devices. In particular, a MEC server and an RF energy transmitter are integrated at the HAP so that it can broadcast wireless energy to devices and execute computational tasks, while each device has a rechargeable battery and an EH circuit component which can store the harvested energy to power its operation. The HAP and all the devices are equipped with a single-antenna\footnote{{To unveil the potential benefits of IRSs in WP-MEC systems for improving the computation rate performance, we assume that the HAP is equipped with a single-antenna for tractability. It is commonly adopted in previous works \cite{8334188, 8434285, 9140412}. The case of multiple antennas is left for our future work.}} and the IRS has $N$ reflecting elements. To ease their practical implementation, all the devices and the HAP operate over the same frequency band. The DL WPT and UL offloading are assumed to operate in time-division multiplexing  manner by following the typical ``harvest-and-then offload'' protocol of \cite{8334188, 8434285}, as shown in Fig. 2. Without loss of generality, we assume that each channel coherence block consists of multiple frames and the operation time of each frame is denoted by $T$. To characterize the maximum achievable computation rate, the channel state information (CSI) of all channels is assumed to be perfectly acquired by the HAP, based on the channel acquisition methods discussed in \cite{8910627}. The channels from the HAP to device $k$, from the HAP to the IRS, and from the IRS to device $k$ are denoted by ${h_{d,k}} \in \mathbb{C}$, ${\bf{g}} \in \mathbb{C}^{N \times 1}$, and ${\bf{h}}_{r,k}^H \in \mathbb{C}^{1 \times N}$, $\forall k \in \left\{ {1, \ldots K} \right\}$, respectively.

In this paper, we assume that the partial computation offloading mode is used. Specifically, the computational tasks of each device can be partitioned into two parts: one for local computing and the other for offloading to the HAP. Similar to \cite{8334188, 8434285, 8986845}, we assume that the local computing at each device adopts a different component from that used for EH and task offloading. Thus, local computing can be executed throughout the entire frame of duration $T$. The number of central processor unit (CPU) cycles required for computing one bit of raw data at each device is denoted by $C$ and its value is determined by the properties of the specific application \cite{8334188}. Let ${f_k}$ denote the CPU's chosen frequency (cycles per second) at device $k$. Therefore, the bits computed locally by device $k$ and the corresponding dissipated energy by local computing are ${{T{f_k}} \mathord{\left/
{\vphantom {{T{f_k}} C}} \right.
 \kern-\nulldelimiterspace} C}$ and $T{\gamma _c}f_k^3$, respectively \cite{8334188}. Note that ${\gamma _c}$ represents the computational energy efficiency of specific CPU chip, which depends on the architecture of the chip \cite{7442079}.
\begin{figure}
\begin{minipage}[t]{0.45\linewidth}
\centering
\includegraphics[width=7cm]{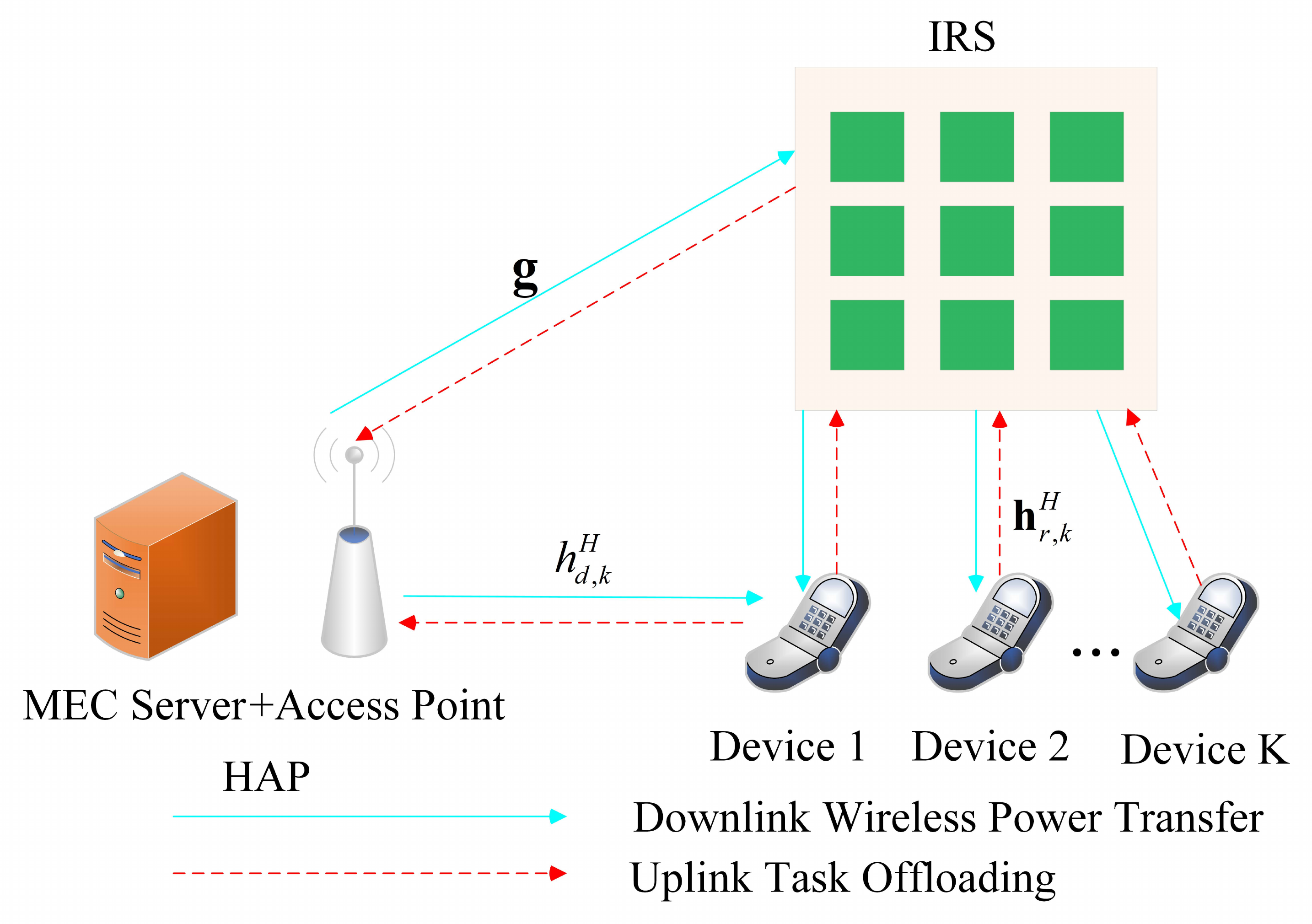}
\vspace{-15pt}\caption{An IRS-aided wireless powered MEC system.}
\label{model}
\end{minipage}%
\hfill
\begin{minipage}[t]{0.45\linewidth}
\centering
\includegraphics[width=7cm]{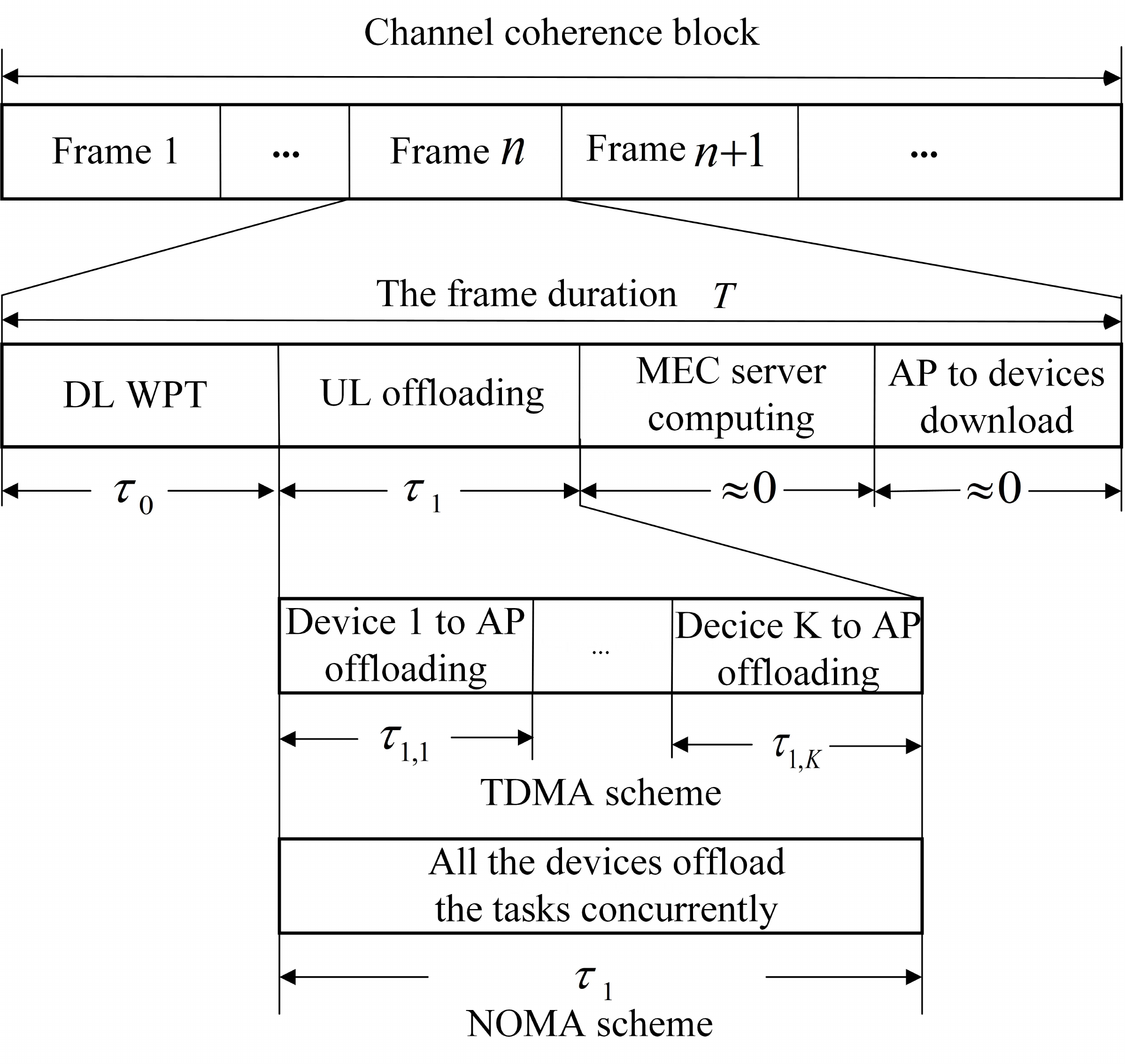}
\vspace{-15pt}\caption{The structure of a transmission frame.}
\label{frame}
\end{minipage}
\vspace{-16pt}
\end{figure}

As shown in Fig. \ref{frame}, the transmission frame is comprised of four segments. First, the HAP broadcasts wireless energy to all devices with the aid of the IRS. Then, all the devices can decide to offload their tasks to the HAP by using TDMA or NOMA. In the third stage, the MEC server at the HAP executes the computational tasks offloaded from all devices. Finally, the computational results are downloaded from the HAP to each device. Since the MEC server has much higher computational capability than those of the devices and the amount of data representing the computational results is negligible \cite{8234686}, the time duration of the third and fourth stages can be neglected as in \cite{8334188, 8434285}. The details of the first and second stages are described as follows.

For DL WPT, an energy signal is broadcasted by the HAP at a constant transmit power ${P_E}$ for a time duration of ${\tau _0}$. The reflection phase-shift matrix of the IRS for DL WPT is denoted by ${{\bf{\Theta }}_0} = {\mathop{\rm diag}\nolimits} \left( {{e^{j{\theta _1}}}, \ldots ,{e^{j{\theta _N}}}} \right)$, where ${\theta _n} \in \left[ {0,2\pi } \right),{\rm{ }}\forall n$. Since the noise power is much lower than the power received from the HAP\cite{6678102}, we assume that the energy harvested from noise is negligible. Based on the  linear EH model\footnote{{Although the non-linear EH model can capture the relationship between the harvested RF power and the converted direct current power more precisely \cite{7843670}, the key results regarding to theoretical performance comparison between NOMA and TDMA for UL offloading are directly applicable to a more general non-linear EH model. It will be discussed later in Remark 3. The linear EH model is adopted here to facilitate us to explicitly demonstrate the impact of IRS on UL offloading activation condition. }} of \cite{8304010, 8334188, 8434285}, the energy harvested at device $k$ is
\begin{align}\label{C1dev}
{E_k} = \eta {\tau _0}{P_E}{\left| {h_{d,k}^H + {\bf{h}}_{r,k}^H{{\bf{\Theta }}_0}{\bf{g}}} \right|^2} = \eta {\tau _0}{P_E}{\left| {h_{d,k}^H + {\bf{q}}_k^H{{\bf{v}}_0}} \right|^2},
\end{align}
where $\eta  \in \left( {0,1} \right]$ represents the energy conversion efficiency of each device, ${\bf{q}}_k^H = {\bf{h}}_{r,k}^H{\mathop{\rm diag}\nolimits} \left( {\bf{g}} \right)$ and ${{\bf{v}}_0} = {\left[ {{e^{j{\theta _1}}}, \ldots ,{e^{j{\theta _N}}}} \right]^T}$ denotes the IRS BF vector of the DL WPT.

At the UL offloading stage, all devices can offload their tasks to the HAP by the TDMA or NOMA schemes. Adopting different IRS BF vectors during the NOMA/TDMA frame, i.e., DIBF, is in principle possible and may potentially improve the computation rate at the cost of additional signaling overhead. This is because the algorithm is usually executed by the HAP due to the limited processing  capability of the IRS and thus the HAP has to feed back the IRS BF vectors to the IRS for reconfiguration. Specifically, we propose three different levels of DIBF schemes as follows: \textbf{Case 1}: The same IRS BF vector is adopted during the entire frame; \textbf{Case 2}: The IRS BF vectors of the DL WPT and UL offloading can be different, but the same IRS BF vector is adopted in the UL offloading stage for all devices; \textbf{Case 3}: The IRS BF vectors of the DL WPT and UL offloading of each device can be different, i.e., $K$ different IRS BF vectors can be used for UL offloading. Considering the aforementioned three cases, the details of UL offloading using TDMA and NOMA are presented as follows.
\subsubsection{Offloading Using TDMA}
The time duration of offloading, namely ${\tau _1}$, is further partitioned into $K$ orthogonal time slots (TSs), which are denoted by ${\tau _{1,k}},\forall k \in \left\{ {1, \ldots K} \right\}$. Device $k$ offloads its data in its $k$-th TS ${\tau _{1,k}}$. Let ${p_k}$ denote the transmit power of device $k$. For \textbf{Case 1}, the DL WPT stage and the UL offloading stage share the same IRS BF vector ${{\bf{v}}_0}$. In this case, the achievable offloading sum-rate is written as
\begin{align}\label{C5dev}
R_{{\rm{off - case1}}}^{{\rm{TDMA}}} = B\sum\limits_{k = 1}^K {\tau _{1,k}{{\log }_2}\left( {1 + \frac{{{p_k}{{\left| {h_{d,k}^H + {\bf{q}}_k^H{{\bf{v}}_0}} \right|}^2}}}{{{\sigma ^2}}}} \right)},
\end{align}
where $B$ represents the system bandwidth and ${{\sigma ^2}}$ denotes the power of the additive white Gaussian noise at the HAP.

For \textbf{Case 2}, we adopt ${{\bf{v}}_1} = {\left[ {{e^{j{\varphi _{1,1}}}}, \ldots ,{e^{j{\varphi _{1,N}}}}} \right]^T}$ to denote the common IRS BF vector in the UL offloading stage. The achievable offloading sum-rate is represented as
\begin{align}\label{C4dev}
R_{{\rm{off - case2}}}^{{\rm{TDMA}}} = B\sum\limits_{k = 1}^K {\tau _{1,k}{{\log }_2}\left( {1 + \frac{{{p_k}{{\left| {h_{d,k}^H + {\bf{q}}_k^H{{\bf{v}}_1}} \right|}^2}}}{{{\sigma ^2}}}} \right)} .
\end{align}

For \textbf{Case 3}, the IRS BF vector used for UL offloading in $k$-th TS is denoted by ${{\bf{v}}_{1,k}} = {\left[ {{e^{j{\varphi _{k,1}}}}, \ldots ,{e^{j{\varphi _{k,N}}}}} \right]^T}$. Thus, the achievable offloading sum-rate is given by
\begin{align}\label{C3dev}
R_{{\rm{off - case3}}}^{{\rm{TDMA}}} = B\sum\limits_{k = 1}^K {{\tau _{1,k}}{{\log }_2}\left( {1 + \frac{{{p_k}{{\left| {h_{d,k}^H + {\bf{q}}_k^H{{\bf{v}}_{1,k}}} \right|}^2}}}{{{\sigma ^2}}}} \right)} .
\end{align}
\vspace{-8pt}
\subsubsection{Offloading Using NOMA}
When NOMA is adopted for UL offloading, all the devices simultaneously transmit their respective data to the HAP throughout the whole time duration of ${\tau _1}$ at the transmit power ${{p_k}}$. To mitigate the inter-user interference, successive interference cancellation (SIC) is performed at the HAP. Taking device $k$ as an example, the HAP will first decode the message of device $i$, $\forall i < k$, before decoding the message of device $k$. Then, the offloading message of device $i$, $\forall i < k$, will be subtracted from the composite signal. The offloading message received from device $i$, $\forall i > k$, is treated as noise. For \textbf{Case 1}, the IRS BF vector is denoted by ${{\bf{v}}_0}$ for the UL offloading stage using NOMA.
Thus, the achievable offloading sum-rate for all the devices is \cite{9400380}
\begin{align}\label{C7dev}
R_{{\rm{off - case1}}}^{{\rm{NOMA}}} = B{\tau _1}{\log _2}\left( {1 + \frac{{\sum\nolimits_{k = 1}^K {{p_k}{{\left| {h_{d,k}^H + {\bf{q}}_k^H{{\bf{v}}_0}} \right|}^2}} }}{{{\sigma ^2}}}} \right).
\end{align}

Accordingly, the achievable offloading sum-rate of \textbf{Case 2} and \textbf{Case 3} for the NOMA-based UL offloading can be written as
\begin{align}\label{C8dev}
&R_{{\rm{off - case2}}}^{{\rm{NOMA}}} = B{\tau _1}{\log _2}\left( {1 \!+\! \frac{{\sum\nolimits_{k = 1}^K {{p_k}{{\left| {h_{d,k}^H \!+\! {\bf{q}}_k^H{{\bf{v}}_1}} \right|}^2}} }}{{{\sigma ^2}}}} \right),\\
&R_{{\rm{off - case3}}}^{{\rm{NOMA}}} = B\sum\limits_{i = 1}^K {{\tau _{1,i}}{{\log }_2}\left( {1 + \frac{{\sum\nolimits_{k = 1}^K {{p_k}{{\left| {h_{d,k}^H + {\bf{q}}_k^H{{\bf{v}}_{1,i}}} \right|}^2}} }}{{{\sigma ^2}}}} \right)} ,
\end{align}
respectively, where ${{\tau _{1,i}}}$ denotes the time duration, when ${{{\bf{v}}_{1,i}}}$ is used.

\subsection{Problem Formulation}
\vspace{-6pt}
In this paper, we aim for maximizing the total number of computed bits of our IRS-aided WP-MEC systems, by jointly optimizing the IRS BF vectors, the time allocation of WPT and offloading, the transmit power of each device, and the local CPU frequency at each device. Both TDMA and NOMA are considered for UL offloading leading to the following formulation:
\subsubsection{TDMA-based Offloading}
When the TDMA scheme is applied, \textbf{Case 1}, \textbf{Case 2}, and \textbf{Case 3} are considered for evaluating the impact of DIBF on the computation rate. The computation rate maximization problem of \textbf{Case 1} can be formulated as\footnote{{To facilitate us to compare the fundamental limits of the achievable computation rate for TDMA and NOMA-based offloading schemes, the quality of service (QoS) constraints of each device are not considered here. Note that the proposed algorithm can be applicable to the scenarios with QoS constraints. Please refer to Remark 5 in Section V for details.}}
\begin{subequations}
\begin{align}
\label{C9-a}({\mathop{\rm P}\nolimits} _{{\rm{TDMA}}}^{{\rm{case1}}}):\mathop {\max }\limits_{{\tau _0},{\tau _{1,k}},{p_k},{{\bf{v}}_0},{f_k}} \;\;&B\sum\limits_{k = 1}^K {\tau _{1,k}{{\log }_2}\left( {1 + \frac{{{p_k}{{\left| {h_{d,k}^H + {\bf{q}}_k^H{{\bf{v}}_0}} \right|}^2}}}{{{\sigma ^2}}}} \right)}  + \sum\limits_{k = 1}^K {\frac{{T{f_k}}}{C}}\\
\label{C9-b}{\rm{s.t.}}\;\;\;\;\;\;\;&\tau _{1,k}{p_k} + T{\gamma _c}f_k^3 \le {\tau _0}\eta {P_E}{\left| {h_{d,k}^H + {\bf{q}}_k^H{{\bf{v}}_0}} \right|^2}{\rm{ }},\;\;\forall k,\\
\label{C9-c}&{\tau _0} + \sum\limits_{k = 1}^K {\tau _{1,k}}  \le T,\\
\label{C9-d}&{\tau _0} \ge 0,{\rm{ }}\tau _{1,k} \ge 0,{\rm{ }}{p_k} \ge 0,{\rm{ }}{f_k} \ge 0{\rm{ }},\;\;\forall k,\\
\label{C9-e}&\left| {{{\left[ {{{\bf{v}}_0}} \right]}_n}} \right| = 1,{\rm{ }}n = 1, \ldots N.
\end{align}
\end{subequations}
In $\left( {{\mathop{\rm P}\nolimits} _{{\rm{TDMA}}}^{{\rm{case1}}}} \right)$,~\eqref{C9-b} represents the energy harvesting causality constraint that the total dissipated energy cannot be higher than the total harvested energy\footnote{{Since we consider that each channel coherence block consists of multiple frames, the consumed energy here comprises two parts. One part is used for UL offloading in the current frame, while the other part is for local computing throughout UL offloading in the current frame and DL WPT in the next frame. Note that constraint \eqref{C9-b} was also adopted in \cite{8234686, 8334188, 8434285}.}} \cite{8234686, 8334188, 8434285}. Furthermore, \eqref{C9-c} is the constraint on the time duration of the DL WPT and UL offloading, while \eqref{C9-d} contains the non-negativity constraints for the optimization variables and \eqref{C9-e} is the unit-modulus constraint for the IRS BF vector. For \textbf{Case 2} and \textbf{Case 3}, the corresponding computation rate maximization problems can be formulated, respectively, as:
\begin{subequations}
\begin{align}
\label{C10-a}({\mathop{\rm P}\nolimits} _{{\rm{TDMA}}}^{{\rm{case2}}}):\mathop {\max }\limits_{{\tau _0},{\tau _{1,k}},{p_k},{{\bf{v}}_0},{{\bf{v}}_1},{f_k}} \;\;&B\sum\limits_{k = 1}^K {\tau _{1,k}{{\log }_2}\left( {1 + \frac{{{p_k}{{\left| {h_{d,k}^H + {\bf{q}}_k^H{{\bf{v}}_1}} \right|}^2}}}{{{\sigma ^2}}}} \right)}  + \sum\limits_{k = 1}^K {\frac{{T{f_k}}}{C}}\\
\label{C10-bb}{\rm{s.t.}}\;\;\;\;\;\;\;\;\;&\left| {{{\left[ {{{\bf{v}}_1}} \right]}_n}} \right| = 1,{\rm{ }}n = 1, \ldots N, \\
\label{C10-c}&\eqref{C9-b}, ~\eqref{C9-c}, ~\eqref{C9-d}, ~\eqref{C9-e},
\end{align}
\end{subequations}

\begin{subequations}
\begin{align}
\label{C11-a}({\mathop{\rm P}\nolimits} _{{\rm{TDMA}}}^{{\rm{case3}}}):\mathop {\max }\limits_{{\tau _0},{\tau _{1,k}},{p_k},{{\bf{v}}_0},{{\bf{v}}_{1,k}},{f_k}} \;\;&B\sum\limits_{k = 1}^K {{\tau _{1,k}}{{\log }_2}\left( {1 + \frac{{{p_k}{{\left| {h_{d,k}^H + {\bf{q}}_k^H{{\bf{v}}_{1,k}}} \right|}^2}}}{{{\sigma ^2}}}} \right)}  + \sum\limits_{k = 1}^K {\frac{{T{f_k}}}{C}}\\
\label{C10-c}{\rm{s.t.}}\;\;\;\;\;\;\;\;\;\;&\left| {{{\left[ {{{\bf{v}}_{1,k}}} \right]}_n}} \right| = 1,n = 1, \ldots N, \;\;\forall k,\\
\label{C10-b}&\eqref{C9-b}, ~\eqref{C9-c}, ~\eqref{C9-d}, ~\eqref{C9-e}.
\end{align}
\end{subequations}
\subsubsection{NOMA-based Offloading}
When NOMA is applied for UL offloading, the corresponding computation rate maximization problems are formulated according to the aforementioned three cases, respectively, as follows:
\begin{subequations}
\begin{align}
\label{C12-a}({\mathop{\rm P}\nolimits} _{{\rm{NOMA}}}^{{\rm{case1}}}):\mathop {\max }\limits_{{\tau _0},{\tau _1},{p_k},{{\bf{v}}_0},{f_k}} \;\;&B{\tau _1}{\log _2}\left( {1 + \frac{{\sum\nolimits_{k = 1}^K {{p_k}{{\left| {h_{d,k}^H + {\bf{q}}_k^H{{\bf{v}}_0}} \right|}^2}} }}{{{\sigma ^2}}}} \right) + \sum\limits_{k = 1}^K {\frac{{T{f_k}}}{C}}\\
\label{C12-b}{\rm{s.t.}}\;\;\;\;\;\;&{\tau _1}{p_k} + T{\gamma _c}f_k^3 \le {\tau _0}\eta {P_E}{\left| {h_{d,k}^H + {\bf{q}}_k^H{{\bf{v}}_0}} \right|^2}{\rm{ }},\;\;\forall k,\\
\label{C12-c}&{\tau _0} + {\tau _1} \le T,\\
\label{C12-d}&{\tau _0} \ge 0,{\tau _1} \ge 0,{\rm{ }}{p_k} \ge 0,{\rm{ }}{f_k} \ge 0{\rm{ }},\;\;\forall k,\\
\label{C12-e}&\eqref{C9-e},
\end{align}
\end{subequations}
\begin{subequations}
\begin{align}
\label{C13-a}({\mathop{\rm P}\nolimits} _{{\rm{NOMA}}}^{{\rm{case2}}}):\mathop {\max }\limits_{{\tau _0},{\tau _1},{p_k},{{\bf{v}}_0},{{\bf{v}}_1},{f_k}} \;\;&B{\tau _1}{\log _2}\left( {1 + \frac{{\sum\nolimits_{k = 1}^K {{p_k}{{\left| {h_{d,k}^H + {\bf{q}}_k^H{{\bf{v}}_1}} \right|}^2}} }}{{{\sigma ^2}}}} \right) + \sum\limits_{k = 1}^K {\frac{{T{f_k}}}{C}}\\
\label{C13-b}{\rm{s.t.}}\;\;\;\;\;\;\;\;&\eqref{C9-e}, ~\eqref{C10-bb}, ~\eqref{C12-b}, ~\eqref{C12-c}, ~\eqref{C12-d},
\end{align}
\end{subequations}
\begin{subequations}
\begin{align}
\label{C100-a}({\mathop{\rm P}\nolimits} _{{\rm{NOMA}}}^{{\rm{case3}}}):\mathop {\max }\limits_{{\tau _0},{\tau _{1,i}},{p_k},{{\bf{v}}_0},{{\bf{v}}_{1,i}},{f_k}} \;&B\sum\limits_{i = 1}^K {{\tau _{1,i}}{{\log }_2}\left( {1 + \frac{{\sum\nolimits_{k = 1}^K {{p_k}{{\left| {h_{d,k}^H + {\bf{q}}_k^H{{\bf{v}}_{1,i}}} \right|}^2}} }}{{{\sigma ^2}}}} \right) + \sum\limits_{k = 1}^K {\frac{{T{f_k}}}{C}} }\\
\label{C100-b}{\rm{s.t.}}\;\;\;\;\;\;\;\;\;&\eqref{C9-c}, ~\eqref{C9-d}, ~\eqref{C9-e}, ~\eqref{C10-c}, ~\eqref{C12-b}.
\end{align}
\end{subequations}
\vspace{-30pt}
\section{TDMA or NOMA for UL Offloading?}
When multiple devices are activated for UL offloading, it still remains unknown which MA scheme is more efficient for UL offloading, especially when considering the impact of the IRS. To answer this question, the theoretical performance comparison between NOMA and TDMA-based UL offloading is provided in this section. First, the impact of DIBF on the computation rate of both NOMA and TDMA-based WP-MEC systems is analyzed. Then, we analytically compare the computation rate achieved by NOMA and TDMA schemes for \textbf{Case 1, 2 and 3}.
\vspace{-8pt}
\subsection{Impact of DIBF on NOMA and TDMA}
For notational simplicity, we use $R_{{\rm{TDMA}}}^{{\rm{case - }}m}$ and $R_{{\rm{NOMA}}}^{{\rm{case }}-m}$ to denote the sum computation rate for \textbf{Case} $m, \left( {m = 1,2,3} \right)$ of TDMA and  NOMA at the optimal solution, respectively. To shed light on the impact of DIBF on the computation rate of the NOMA and TDMA schemes, we first introduce the following lemmas.
\begin{lem}\label{lem}
For IRS-aided WP-MEC systems employing NOMA for offloading, it follows that $R_{{\rm{NOMA}}}^{{\rm{case1}}} \le R_{{\rm{NOMA}}}^{{\rm{case2}}} = R_{{\rm{NOMA}}}^{{\rm{case3}}}$.
\end{lem}
\begin{proof}
Assume that an optimal solution of $({\mathop{\rm P}\nolimits} _{{\rm{NOMA}}}^{{\rm{case3}}})$ is given by $\left\{ {\tau _0^*,\tau _{1,i}^*,p_k^*,f_k^*,{\bf{v}}_0^*,{\bf{v}}_{1,i}^*} \right\}$. Then, the optimal value of $({\mathop{\rm P}\nolimits} _{{\rm{NOMA}}}^{{\rm{case3}}})$ can be expressed as
\begin{align}\label{optimal_value_noma3}
R_{{\rm{NOMA}}}^{{\rm{case3}}} = \sum\limits_{i = 1}^K {\tau _{1,i}^*{{\log }_2}\left( {1 + \frac{{\sum\nolimits_{k = 1}^K {p_k^*{{\left| {h_{d,k}^H + {\bf{q}}_k^H{\bf{v}}_{1,i}^*} \right|}^2}} }}{{{\sigma ^2}}}} \right) + \sum\limits_{k = 1}^K {\frac{{Tf_k^*}}{C}} }.
\end{align}
There always exists an IRS BF vector denoted by ${\bf{v}}_{1,p}^*$, $p \in \left\{ {1, \cdots ,K} \right\}$, which satisfies ${\bf{v}}_{1,p}^* = \mathop {\arg \max }\limits_{{\bf{v}}_{1,i}^*,i \in \left\{ {1, \cdots ,K} \right\}} \sum\limits_{k = 1}^K {p_k^*{{\left| {h_{d,k}^H + {\bf{q}}_k^H{\bf{v}}_{1,i}^*} \right|}^2}} $. As such, we have
\begin{align}\label{noma3_vs_noma2}
R_{{\rm{NOMA}}}^{{\rm{case3}}} & \le \sum\limits_{i = 1}^K {\tau _{1,i}^*{{\log }_2}\left( {1 + \frac{{\sum\nolimits_{k = 1}^K {p_k^*{{\left| {h_{d,k}^H + {\bf{q}}_k^H{\bf{v}}_{1,p}^*} \right|}^2}} }}{{{\sigma ^2}}}} \right) + \sum\limits_{k = 1}^K {\frac{{Tf_k^*}}{C}} } \nonumber\\
& = {{\tilde \tau }_1}{\log _2}\left( {1 + \frac{{\sum\nolimits_{k = 1}^K {p_k^*{{\left| {h_{d,k}^H + {\bf{q}}_k^H{\bf{v}}_{1,p}^*} \right|}^2}} }}{{{\sigma ^2}}}} \right) + \sum\limits_{k = 1}^K {\frac{{Tf_k^*}}{C}}  \le R_{{\rm{NOMA}}}^{{\rm{case2}}},
\end{align}
where ${{\tilde \tau }_1} = \sum\limits_{i = 1}^K {\tau _{1,i}^*} $. The equality holds if ${\bf{v}}_{1,i}^* = {\bf{v}}_{1,p}^*,\forall i$. Meanwhile, by setting ${{\bf{v}}_{1,i}} = {{\bf{v}}_{1,j}},\forall i,j$, problem $({\mathop{\rm P}\nolimits} _{{\rm{NOMA}}}^{{\rm{case3}}})$ is reduced to $({\mathop{\rm P}\nolimits} _{{\rm{NOMA}}}^{{\rm{case2}}})$, which yields $R_{{\rm{NOMA}}}^{{\rm{case2}}} \le R_{{\rm{NOMA}}}^{{\rm{case3}}}$. Thus, we have $R_{{\rm{NOMA}}}^{{\rm{case2}}} = R_{{\rm{NOMA}}}^{{\rm{case3}}}$. Similarly, problem $({\mathop{\rm P}\nolimits} _{{\rm{NOMA}}}^{{\rm{case2}}})$ is reduced to $({\mathop{\rm P}\nolimits} _{{\rm{NOMA}}}^{{\rm{case1}}})$ by setting ${{\bf{v}}_1} = {{\bf{v}}_0}$, thus we have $R_{{\rm{NOMA}}}^{{\rm{case1}}} \le R_{{\rm{NOMA}}}^{{\rm{case2}}}$.
\end{proof}

\begin{lem}\label{lem}
For IRS-aided WP-MEC systems employing TDMA for offloading, it follows that $R_{{\rm{TDMA}}}^{{\rm{case1}}} \le R_{{\rm{TDMA}}}^{{\rm{case2}}} \le R_{{\rm{TDMA}}}^{{\rm{case3}}}$.
\end{lem}
\begin{proof}
By setting ${{\bf{v}}_1} = {{\bf{v}}_0}$, problem $({\mathop{\rm P}\nolimits} _{{\rm{TDMA}}}^{{\rm{case2}}})$ is reduced to $({\mathop{\rm P}\nolimits} _{{\rm{TDMA}}}^{{\rm{case1}}})$, which yields $R_{{\rm{TDMA}}}^{{\rm{case1}}} \le R_{{\rm{TDMA}}}^{{\rm{case2}}}$. For \textbf{Case 3}, the equivalent channel power gain of each device can be maximized by setting ${\bf{v}}_{1,k}$ to align the cascaded link with the direct link ${h_{d,k}^H}$. Thus, ${\left| {h_{d,k}^H + {\bf{q}}_k^H{\bf{v}}_{1,k}^*} \right|^2} \ge {\left| {h_{d,k}^H + {\bf{q}}_k^H{\bf{v}}_1^*} \right|^2}$ holds for device $k,\forall k$, which yields $R_{{\rm{TDMA}}}^{{\rm{case2}}} \le R_{{\rm{TDMA}}}^{{\rm{case3}}}$.
\end{proof}

Lemma 1 and Lemma 2 provide the following insights and also serve as the theoretical foundation for comparing TDMA and NOMA-based offloading, which will be discussed later.
\begin{itemize}
  \item For NOMA-based UL offloading, varying the IRS BF vectors in the UL offloading stage does not necessarily attain performance improvements over a static IRS BF vector. By contrast, for TDMA-based UL offloading, the computation rate can be further improved by varying IRS BF vectors for UL offloading.
  \item For both TDMA and NOMA-based WP-MEC systems, having different IRS BF vectors for the DL WPT and UL offloading generally outperforms its counterpart using the same IRS BF vector throughout the entire frame.
\end{itemize}

\vspace{-8pt}
\subsection{TDMA versus NOMA-based UL Offloading}

To compare the achievable computation rate performance between offloading using TDMA and  NOMA, the relationship between $({\mathop{\rm P}\nolimits} _{{\rm{TDMA}}}^{{\rm{case2}}})$ and $({\mathop{\rm P}\nolimits} _{{\rm{NOMA}}}^{{\rm{case2}}})$ is presented in the following theorem.
\begin{thm}\label{thm2}
Assuming that $\left\{ {\tau _0^*,\tau _{1,k}^*,p_k^*,{\bf{v}}_0^*,{\bf{v}}_1^*,f_k^*} \right\}$ and $\left\{ {\tau _0^ \star ,\tau _1^ \star ,p_k^ \star ,{\bf{v}}_0^ \star ,{\bf{v}}_1^ \star ,f_k^ \star } \right\}$ are the optimal solutions of $({\mathop{\rm P}\nolimits} _{{\rm{TDMA}}}^{{\rm{case2}}})$ and $({\mathop{\rm P}\nolimits} _{{\rm{NOMA}}}^{{\rm{case2}}})$, respectively, we have $R_{{\rm{TDMA}}}^{{\rm{case2}}} = R_{{\rm{NOMA}}}^{{\rm{case2}}}$ with $\tau _0^ \star  = \tau _0^*,\tau _1^ \star  = \sum\nolimits_{k = 1}^K {\tau _{1,k}^*} ,p_k^ \star  = p_k^*,{\bf{v}}_0^ \star  = {\bf{v}}_0^*,{\bf{v}}_1^ \star  = {\bf{v}}_1^*$ and $f_k^ \star  = f_k^*$.
\end{thm}
\begin{proof}
See Appendix A.
\end{proof}

Note that the similar results presented in Theorem 1 can be directly extended to capture the  interrelation between $({\mathop{\rm P}\nolimits} _{{\rm{TDMA}}}^{{\rm{case1}}})$ and $({\mathop{\rm P}\nolimits} _{{\rm{NOMA}}}^{{\rm{case1}}})$, i.e., $R_{{\rm{TDMA}}}^{{\rm{case1}}} = R_{{\rm{NOMA}}}^{{\rm{case1}}}$. Theorem 1 explicitly shows that the solutions of problem $({\mathop{\rm P}\nolimits} _{{\rm{NOMA}}}^{{\rm{case1}}})$ and $({\mathop{\rm P}\nolimits} _{{\rm{NOMA}}}^{{\rm{case2}}})$ can be directly obtained based on those of $({\mathop{\rm P}\nolimits} _{{\rm{TDMA}}}^{{\rm{case1}}})$ and $({\mathop{\rm P}\nolimits} _{{\rm{TDMA}}}^{{\rm{case2}}})$, respectively.
\begin{rem}
The results presented in Lemma 1, Lemma 2, and Theorem 1 answer the fundamental question regarding the computation rate comparison between offloading using TDMA and NOMA. Specifically, the comparison outcome depends on which DIBF scheme is applied. For \textbf{Case 1} and \textbf{Case 2}, it is shown that the same computation rate can be achieved by using TDMA and NOMA for offloading. Since the computation rate of TDMA can be further improved by adapting IRS BF vectors over different TSs in the UL offloading stage, the computation rate of TDMA becomes higher than that of NOMA for \textbf{Case 3} at the cost of extra signaling overhead. As such, we have the inequality chain as follows:
\begin{align}\label{noma_verus_tdma}
R_{{\rm{TDMA}}}^{{\rm{case1}}} = R_{{\rm{NOMA}}}^{{\rm{case1}}} \le R_{{\rm{TDMA}}}^{{\rm{case2}}} = R_{{\rm{NOMA}}}^{{\rm{case2}}} = R_{{\rm{NOMA}}}^{{\rm{case3}}} \le R_{{\rm{TDMA}}}^{{\rm{case3}}}.
\end{align}
\end{rem}
\begin{rem}
Considering the high $C$ regime, i.e., $C \to  + \infty$, which implies that the device has nearly no computing capability to deal with computationally intensive tasks, the computations completely rely on offloading the tasks to MEC servers. In this case, the computation rate maximization problem is equivalent to the throughput maximization problem of WPCNs. For \textbf{Case 2}, our previous work \cite{9400380} unveiled that the same IRS BF vector can be exploited for the DL and UL in WPCNs without loss of optimality, i.e., ${\bf{v}}_0^* = {\bf{v}}_1^*$. Based on the results provided in Theorem 1, we have the following relationship in the high $C$ regime:
\begin{align}\label{high_C_regime}
R_{{\rm{TDMA}}}^{{\rm{case1}}} = R_{{\rm{NOMA}}}^{{\rm{case1}}} = R_{{\rm{TDMA}}}^{{\rm{case2}}} = R_{{\rm{NOMA}}}^{{\rm{case2}}} = R_{{\rm{NOMA}}}^{{\rm{case3}}} \le R_{{\rm{TDMA}}}^{{\rm{case3}}}.
\end{align}
\end{rem}
In contrast to \eqref{noma_verus_tdma}, \eqref{high_C_regime} suggests that for \textbf{Case 2}, DL WPT and UL offloading can adopt the same IRS BF vector without loss of optimality at a lower signaling overhead.
\begin{rem}
Note that the theoretical comparison provided in Remark 2 can be directly extended to a more general non-linear EH model. For a general EH model, the output direct current power can be generally expressed as a function of the input RF power, i.e, ${Q_k}\left( {{P_E}{{\left| {h_{d,k}^H + {\bf{q}}_k^H{\bf{v}}_0^*} \right|}^2}} \right)$. Replacing $\eta {P_E}{\left| {h_{d,k}^H + {\bf{q}}_k^H{\bf{v}}_0^*} \right|^2}$ by ${Q_k}\left( {{P_E}{{\left| {h_{d,k}^H + {\bf{q}}_k^H{\bf{v}}_0^*} \right|}^2}} \right)$ in Appendix A, the results can be directly obtained through similar steps.
\end{rem}
\vspace{-12pt}
\section{UL Offloading activation Condition in single-user Systems}
\vspace{-4pt}
Before deriving the solutions of the aforementioned computation rate maximization problems, we consider the special case of a single-user setup, i.e., $K = 1$, to gain important insights into the efficiency of IRSs for UL offloading activation. In this case, the MA schemes have no impact on the results, and thus the computation rate maximization problems are simplified to:
\begin{subequations}\label{SU}
\begin{align}
\mathop {\max }\limits_{{\tau _0},{\tau _1},p,{\bf{v}},f} \;\;&B{\tau _1}{\log _2}\left( {1 + \frac{{p{{\left| {h_d^H + {{\bf{q}}^H}{\bf{v}}} \right|}^2}}}{{{\sigma ^2}}}} \right) + \frac{{Tf}}{C}\\
\label{SU-b}{\rm{s.t.}}\;\;\;\;\;&{\tau _1}p + T{\gamma _c}{f^3} \le {\tau _0}\eta {P_E}{\left| {h_d^H + {{\bf{q}}^H}{\bf{v}}} \right|^2},\\
\label{SU-c}&{\tau _0} + {\tau _1} \le T,\\
\label{SU-d}&{\tau _0} \ge 0,{\tau _1} \ge 0,p \ge 0,f \ge 0,\\
\label{SU-e}&\left| {{{\left[ {\bf{v}} \right]}_n}} \right| = 1,n = 1, \ldots N.
\end{align}
\end{subequations}
Problem \eqref{SU} has not been investigated in previous articles to the best of our knowledge. Note that for a WP-MEC system, UL offloading may not be activated, when suffering from severe wireless channel conditions. Hence, we focus our attention on a single-user case to unveil the impact of IRSs on the UL offloading activation condition. For problem \eqref{SU}, the optimal IRS BF vector ${\bf{v}}$ can be directly obtained as ${\left[ {{{\bf{v}}^*}} \right]_n} = {e^{j\left( {\left( {\arg \left( {h_d^H} \right) + \arg \left( {{{\left[ {\bf{q}} \right]}_n}} \right)} \right)} \right)}}$, which aligns the cascaded channel between a typical device and the HAP via the IRS with the end-to-end channel. By setting ${\bf{v}}$ as the optimal form, the channel power gain between a typical device and HAP is determined. Let $h = {\left| {h_d^H + {{\bf{q}}^H}{\bf{v}^*}} \right|^2}$ for notational simplicity. Then, problem \eqref{SU} can be further transformed into a resource allocation optimization problem (OP) as follows
\begin{subequations}\label{SURA}
\begin{align}
\label{SURA-a}\mathop {\max }\limits_{{\tau _0},{\tau _1},e,f} \;\;&B{\tau _1}{\log _2}\left( {1 + \frac{{e{h^2}}}{{{\tau _1}{\sigma ^2}}}} \right) + \frac{{Tf}}{C}\\
\label{SURA-b}{\rm{s.t.}}\;\;\;&e + T{\gamma _c}{f^3} \le {\tau _0}\eta {P_E}h, \\
\label{SURA-c}&{\tau _0} \ge 0,{\tau _1} \ge 0,e \ge 0,f \ge 0, \\
\label{SURA-d}&\eqref{SU-c},
\end{align}
\end{subequations}
where $e = {\tau _1}p$. It may be readily shown that problem \eqref{SURA} is a convex OP. By analyzing the KKT conditions of problem \eqref{SURA}, the general UL offloading activation condition admitting a threshold-based structure is obtained in the following proposition.
\begin{proposition}
For the single-user setup, UL offloading for the typical device will be activated if and only if the following condition is satisfied,
\begin{align}\label{C30dev}
{P_E} > {\rm{thre}}\left( h \right) = {\gamma _c}\frac{1}{{\eta h}}{\left( {\frac{{\left( {{\sigma ^2} + {p^*}h} \right)\ln 2}}{{3Ch{\gamma _c}B}}} \right)^{\frac{3}{2}}},
\end{align}
where ${p^*}$ is the unique solution of
\begin{align}\label{C31dev}
G\left( {p,h} \right) \buildrel \Delta \over = {\log _2}\left( {1 + \frac{{ph}}{{{\sigma ^2}}}} \right) - \frac{{ph}}{{\left( {{\sigma ^2} + ph} \right)\ln 2}} - \eta {P_E}\frac{{{h^2}}}{{\left( {{\sigma ^2} + ph} \right)\ln 2}} = 0.
\end{align}
\end{proposition}
\begin{proof}
Please see Appendix B.
\end{proof}

Proposition 1 explicitly shows that a typical device would prefer UL offloading for maximizing its computation rate, when the transmit power of the HAP is higher than a threshold ${\mathop{\rm thre}\nolimits} \left( h \right)$, which depends on the channel power gain $h$. Based on Proposition 1, we further discuss the monotonic relationship between ${\mathop{\rm thre}\nolimits} \left( h \right)$ and $h$.
\begin{proposition}
The threshold ${\mathop{\rm thre}\nolimits} \left( h \right)$ decreases with the equivalent channel power gain $h$.
\end{proposition}
\begin{proof}
Since $\ln \left( {{\rm{thre}}\left( h \right)} \right)$ has the same monotonic relationship with $h$ as that of ${{\rm{thre}}\left( h \right)}$, we focus our attention on showing that $\ln \left( {{\rm{thre}}\left( h \right)} \right)$ decreases with $h$ instead. Taking the first order derivative of $\ln \left( {{\rm{thre}}\left( h \right)} \right)$ with respect to $h$, we obtain
\begin{align}\label{C32dev}
\frac{{\partial \ln \left( {{\rm{thre}}\left( h \right)} \right)}}{{\partial h}}{\rm{ = }}\frac{{\partial \left( { - \frac{5}{2}\ln h + \frac{3}{2}\ln \left( {{\sigma ^2} + h{p^*}(h)} \right)} \right)}}{{\partial h}} =  - \frac{5}{2}\frac{1}{h} + \frac{3}{2}\frac{{\bar p + h\frac{{d{p^*}(h)}}{{dh}}}}{{{\sigma ^2} + h{p^*}(h)}}.
\end{align}
Note that ${p^*}$ is a function of $h$, which is determined by \eqref{C31dev}. As such, we use ${p^*}\left( h \right)$ instead of ${p^*}$ in the following. Based on the method of implicit differentiation, we obtain
\begin{align}\label{p_h}
\frac{{\partial {p^*}(h)}}{{\partial h}} =  - \frac{{{{\partial G\left( {p,h} \right)} \mathord{\left/
 {\vphantom {{\partial G\left( {p,h} \right)} {\partial h}}} \right.
 \kern-\nulldelimiterspace} {\partial h}}}}{{{{\partial G\left( {p,h} \right)} \mathord{\left/
 {\vphantom {{\partial G\left( {p,h} \right)} {\partial p}}} \right.
 \kern-\nulldelimiterspace} {\partial p}}}} = \frac{{\eta {P_E}\left( {2{\sigma ^2} + h{p^*}(h)} \right) - {{\left( {{p^*}(h)} \right)}^2}}}{{h\left( {\eta {P_E}h + {p^*}(h)} \right)}}.
\end{align}
Substituting \eqref{p_h} into \eqref{C32dev} yields
\begin{align}\label{relation_th_h}
\frac{{\partial \ln \left( {{\rm{thre}}\left( h \right)} \right)}}{{\partial h}} = \frac{{ - \frac{5}{2}{p^*}(h) - \frac{1}{2}\eta {P_E}h}}{{h\left( {{p^*}(h) + \eta {P_E}h} \right)}} < 0.
\end{align}
Thus, ${\mathop{\rm thre}\nolimits} \left( h \right)$ decreases with $h$.
\end{proof}

Proposition 1 and  Proposition 2 serve as a solid theoretical foundation for further investigating the impact of IRSs on the UL offloading activation condition. For ease of exposition, we assume that the IRS can establish pure line-of-sight (LoS) links with both the device and the AP. By setting the IRS BF vector as the optimal form, the equivalent channel power gain can be formulated as
\begin{align}\label{SU_gain}
h = {\left| {h_d^H + {{\bf{q}}^H}{{\bf{v}}^*}} \right|^2} = \beta d_{{\rm{AD}}}^{ - {\alpha _{{\rm{AD}}}}}{\left( {1 + N\sqrt {d_{{\rm{AD}}}^{{\alpha _{{\rm{AD}}}}}d_{{\rm{AI}}}^{ - {\alpha _{{\rm{AI}}}}}} \sqrt {\beta d _{{\rm{ID}}}^{ - {\alpha _{{\rm{ID}}}}}} } \right)^2},
\end{align}
where ${d_{{\rm{AD}}}}$ (${\alpha _{{\rm{AD}}}}$), ${d_{{\rm{AI}}}}$ (${\alpha _{{\rm{AI}}}}$), and ${d_{{\rm{ID}}}}$ (${\alpha _{{\rm{ID}}}}$) denote the length (path-loss exponent) of the HAP-device, HAP-IRS, and IRS-device links, respectively, and $\beta$ represents the channel power gain at a reference distance of 1 meter (m).
\begin{rem}
For a specific dominant LoS scenario, the UL offloading activation condition can be expressed as
\begin{align}\label{C30dev}
{P_E} > \frac{{{\gamma _c}}}{{\eta \beta d_{{\rm{AD}}}^{ - {\alpha _{{\rm{AD}}}}}}}{\left( {\frac{{\left( {{\sigma ^2} + {p^*}{{\left( {1 + N\sqrt {d_{{\rm{AD}}}^{{\alpha _{{\rm{AD}}}}}d_{{\rm{AI}}}^{ - {\alpha _{{\rm{AI}}}}}} \sqrt {\beta d_{{\rm{ID}}}^{ - {\alpha _{{\rm{ID}}}}}} } \right)}^2}} \right)\ln 2}}{{3C{{\left( {1 + N\sqrt {d_{{\rm{AD}}}^{{\alpha _{{\rm{AD}}}}}d_{{\rm{AI}}}^{ - {\alpha _{{\rm{AI}}}}}} \sqrt {\beta d_{{\rm{ID}}}^{ - {\alpha _{{\rm{ID}}}}}} } \right)}^{\frac{{10}}{3}}}{\gamma _c}B}}} \right)^{\frac{3}{2}}}.
\end{align}
It is plausible that the value of $h$ using IRSs becomes ${\left( {1 + N\sqrt {d_{{\rm{AD}}}^{{\alpha _{{\rm{AD}}}}}d_{{\rm{AI}}}^{ - {\alpha _{{\rm{AI}}}}}} \sqrt {\beta d _{{\rm{ID}}}^{ - {\alpha _{{\rm{ID}}}}}} } \right)^2}$ times higher than that without IRSs. By increasing the number of IRS elements $N$, the channel power gain $h$ can be significantly increased, which substantially reduces the threshold ${{\rm{thre}}\left( h \right)}$ for UL offloading. Thus, a typical device is more willing to perform task offloading upon increasing of $N$ due to the improved channel conditions. This confirms the practicality of deploying IRSs in next generation communication networks.
\end{rem}

\vspace{-8pt}
\section{Proposed Solutions for General Multi-user systems}
In this section, we focus our attention on solving the computation rate maximization problems of TDMA-based UL offloading, i.e., $\left( {{\mathop{\rm P}\nolimits} _{{\rm{TDMA}}}^{{\rm{case1}}}} \right)$, $\left( {{\mathop{\rm P}\nolimits} _{{\rm{TDMA}}}^{{\rm{case2}}}} \right)$, and $\left( {{\mathop{\rm P}\nolimits} _{{\rm{TDMA}}}^{{\rm{case3}}}} \right)$. Solving the same problems for NOMA is similar to those of TDMA according to the results of Section III.
\vspace{-8pt}
\subsection{AO Algorithm Proposed for Solving $\left( {{\mathop{\rm P}\nolimits} _{{\rm{TDMA}}}^{{\rm{case1}}}} \right)$ and $\left( {{\mathop{\rm P}\nolimits} _{{\rm{TDMA}}}^{{\rm{case2}}}} \right)$}
Since problem $\left( {{\mathop{\rm P}\nolimits} _{{\rm{TDMA}}}^{{\rm{case2}}}} \right)$ is more complex than $\left( {{\mathop{\rm P}\nolimits} _{{\rm{TDMA}}}^{{\rm{case1}}}} \right)$, we commence with \textbf{Case 2}, i.e., $\left( {{\mathop{\rm P}\nolimits} _{{\rm{TDMA}}}^{{\rm{case2}}}} \right)$. It will be shown later in this section that an algorithm designed for solving $\left( {{\mathop{\rm P}\nolimits} _{{\rm{TDMA}}}^{{\rm{case2}}}} \right)$ may also be directly applicable to $\left( {{\mathop{\rm P}\nolimits} _{{\rm{TDMA}}}^{{\rm{case1}}}} \right)$. For problem $\left( {{\mathop{\rm P}\nolimits} _{{\rm{TDMA}}}^{{\rm{case2}}}} \right)$, the optimization variable ${p_k}$ is closely coupled with the variables ${\tau _{1,k}}$ and ${{{\bf{v}}_1}}$, while ${\tau _0}$ is coupled with ${{{\bf{v}}_0}}$. Moreover, the unit-modulus constraints in~\eqref{C9-e} and~\eqref{C10-bb} render problem $\left( {{\mathop{\rm P}\nolimits} _{{\rm{TDMA}}}^{{\rm{case2}}}} \right)$ non-convex. In order to deal with the closely-coupled non-convex terms in problem $\left( {{\mathop{\rm P}\nolimits} _{{\rm{TDMA}}}^{{\rm{case2}}}} \right)$, we decompose the original problem into a pair of subproblems. Specifically, the resource allocation OP with respect to $\left\{ {{\tau _0},{\tau _{1,k}},{p_k},{f_k}} \right\}$ and the IRS BF OP with respect to $\left\{ {{{\bf{v}}_0},{{\bf{v}}_1}} \right\}$ can be efficiently solved in an alternating manner as described next.
\subsubsection{Resource Allocation Optimization}
Under any given feasible IRS BF vectors ${{{\bf{v}}_0}}$ and ${{{\bf{v}}_1}}$, the resource allocation OP with respect to $\left\{ {{\tau _0},{\tau _{1,k}},{p_k},{f_k}} \right\}$ may be written as
\begin{subequations}\label{C14}
\begin{align}
\label{C14-a}\mathop {\max }\limits_{{\tau _0},{\tau _{1,k}},{p_k},{f_k}} \;\;&B\sum\limits_{k = 1}^K {{\tau _{1,k}}{{\log }_2}\left( {1 + \frac{{{p_k}g_k^{{\rm{off}}}}}{{{\sigma ^2}}}} \right)}  + \sum\limits_{k = 1}^K {\frac{{T{f_k}}}{C}}\\
\label{C14-b}{\rm{s.t.}}\;\;\;\;\;\;&{\tau _{1,k}}{p_k} + T{\gamma _c}f_k^3 \le {\tau _0}\eta {P_E}h_k^{{\rm{wpt}}}\;\;\forall k,\\
\label{C14-c}&\eqref{C9-c}, ~\eqref{C9-d}.
\end{align}
\end{subequations}
Here, we use $h_k^{{\rm{wpt}}} = {\left| {h_{d,k}^H + {\bf{q}}_k^H{{\bf{v}}_0}} \right|^2}$ and $g_k^{{\rm{off}}} = {\left| {h_{d,k}^H + {\bf{q}}_k^H{{\bf{v}}_1}} \right|^2}$ to denote the equivalent channel power gains of device $k$ for the DL WPT and UL offloading links, respectively. Observe that \eqref{C14} is still a non-convex OP due the coupled variables ${\tau _{1,k}}$, ${p_k}$ in~\eqref{C14-b} and the non-concave objective function. To tackle this issue, problem \eqref{C14} can be equivalently transformed into the following OP by letting ${e_k} = {\tau _{1,k}}{p_k}$
\begin{subequations}\label{C15}
\begin{align}
\label{C15-a}\mathop {\max }\limits_{{\tau _0},{\tau _{1,k}},{e_k},{f_k}} \;\;&B\sum\limits_{k = 1}^K {{\tau _{1,k}}{{\log }_2}\left( {1 + \frac{{{e_k}g_k^{{\rm{off}}}}}{{{\tau _{1,k}}{\sigma ^2}}}} \right)}  + \sum\limits_{k = 1}^K {\frac{{T{f_k}}}{C}}\\
\label{C15-b}{\rm{s.t.}}\;\;\;\;\;\;&{e_k} + T{\gamma _c}f_k^3 \le {\tau _0}\eta {P_E}h_k^{\rm{wpt }}\;\;\forall k,\\
\label{C15-c}&{\tau _0} \ge 0,{\tau _{1,k}} \ge 0,{\rm{ }}{e_k} \ge 0,{\rm{ }}{f_k} \ge 0{\rm{ }}\;\;\forall k,\\
\label{C15-d}&\eqref{C9-c}.
\end{align}
\end{subequations}
It can be verified that problem \eqref{C15} is a convex OP whose optimal solutions can be efficiently obtained by standard numerical methods, e.g., the interior point method \cite{convex}.
\subsubsection{Optimization of IRS BF Vectors}
For any given feasible set $\left\{ {{\tau _0},{\tau _{1,k}},{p_k},{f_k}} \right\}$, the OP with respect to $\left\{ {{{\bf{v}}_0},{{\bf{v}}_1}} \right\}$ can be written as
\begin{subequations}\label{C44}
\begin{align}
\label{C44-a}\mathop {\max {\rm{ }}}\limits_{{{\bf{v}}_0},{{\bf{v}}_1}} \;\;&B\sum\limits_{k = 1}^K {{\tau _{1,k}}{{\log }_2}\left( {1 + \frac{{{p_k}{{\left| {h_{d,k}^H + {\bf{q}}_k^H{{\bf{v}}_1}} \right|}^2}}}{{{\sigma ^2}}}} \right)}  + \sum\limits_{k = 1}^K {\frac{{T{f_k}}}{C}}\\
\label{C44-b}{\rm{s.t.}}\;\;\;&\eqref{C9-b},~\eqref{C9-e},~\eqref{C10-bb}.
\end{align}
\end{subequations}
Problem \eqref{C44} is challenging to solve due to the non-concave objective function and non-convex constraints in~\eqref{C44-b}. To make problem \eqref{C44} more tractable, we first relax the unit-modulus constraints~\eqref{C9-e} and~\eqref{C10-bb} into $\left| {{{\left[ {{{\bf{v}}_0}} \right]}_n}} \right| \le 1$ and $\left| {{{\left[ {{{\bf{v}}_1}} \right]}_n}} \right| \le 1$, which yields the following problem:
\begin{subequations}\label{IRS}
\begin{align}
\label{IRS-a}\mathop {\max {\rm{ }}}\limits_{{{\bf{v}}_0},{{\bf{v}}_1}} \;\;&B\sum\limits_{k = 1}^K {{\tau _{1,k}}{{\log }_2}\left( {1 + \frac{{{p_k}{{\left| {h_{d,k}^H + {\bf{q}}_k^H{{\bf{v}}_1}} \right|}^2}}}{{{\sigma ^2}}}} \right)}  + \sum\limits_{k = 1}^K {\frac{{T{f_k}}}{C}}\\
\label{IRS-b}{\rm{s.t.}}\;\;\;&\left| {{{\left[ {{{\bf{v}}_0}} \right]}_n}} \right| \le 1,{\rm{ }}n = 1, \;\;\ldots N, \\
\label{IRS-c}&\left| {{{\left[ {{{\bf{v}}_1}} \right]}_n}} \right| \le 1,{\rm{ }}n = 1, \;\;\ldots N, \\
\label{IRS-d}&\eqref{C9-b}.
\end{align}
\end{subequations}
In the following, we focus our attention on the IRS BF vector OP relying on continuous amplitudes. To handle the non-concave objective function, we introduce a slack variable ${S_k}$ and reformulate problem \eqref{IRS} as follows
\begin{subequations}\label{C45}
\begin{align}
\label{C45-a}\mathop {\max {\rm{ }}}\limits_{{S_k},{{\bf{v}}_0},{{\bf{v}}_1}} \;&B\sum\limits_{k = 1}^K {{\tau _{1,k}}{{\log }_2}\left( {1 + \frac{{{p_k}{S_k}}}{{{\sigma ^2}}}} \right)}  + \sum\limits_{k = 1}^K {\frac{{T{f_k}}}{C}}\\
\label{C45-b}{\rm{s.t.}}\;\;\;\;&\eqref{C9-b},~\eqref{IRS-b},~\eqref{IRS-c}, \\
\label{C45-c}&{S_k} \le {\left| {h_{d,k}^H + {\bf{q}}_k^H{{\bf{v}}_1}} \right|^2}{\rm{ }},\;\;\forall k.
\end{align}
\end{subequations}
Note that problem \eqref{C45} is equivalent to problem \eqref{IRS} since constraint~\eqref{C45-c} is satisfied with equality at the optimal solution. Problem \eqref{C45} is still a non-convex OP due to constraints~\eqref{C45-b} and~\eqref{C45-c}. However, the convexity of ${\left| {h_{d,k}^H + {\bf{q}}_k^H{{\bf{v}}_0}} \right|^2}$ and ${\left| {h_{d,k}^H + {\bf{q}}_k^H{{\bf{v}}_1}} \right|^2}$ allows us to apply the SCA technique to deal with constraints \eqref{C9-b} and \eqref{C45-c}. Let ${\left| {h_{d,k}^H + {\bf{q}}_k^H{{\bf{v}}_0}} \right|^2} = {\left| {{\bf{\bar q}}_k^H{{{\bf{\bar v}}}_0}} \right|^2}$ and ${\left| {h_{d,k}^H + {\bf{q}}_k^H{{\bf{v}}_1}} \right|^2} = {\left| {{\bf{\bar q}}_k^H{{{\bf{\bar v}}}_1}} \right|^2}$, where ${{{\bf{\bar v}}}_0} = {\left[ {{\bf{\bar v}}_0^H,1} \right]^H}$, ${{{\bf{\bar v}}}_1} = {\left[ {{\bf{\bar v}}_1^H,1} \right]^H}$, and ${\bf{\bar q}}_k^H = \left[ {{\bf{q}}_k^H,h_{d,k}^H} \right]$. Specifically, we use $l$ $\left( {l \in {\mathbb{Z}^ + }} \right)$ to denote the iteration index. At the $l$-th iteration, where a given local point is denoted by $\left\{ {{\bf{\bar v}}_0^{\left( l \right)},{\bf{\bar v}}_1^{\left( l \right)}} \right\}$, we have
\begin{align}\label{C46-1dev}
{{{\left| {h_{d,k}^H + {\bf{q}}_k^H{{\bf{v}}_0}} \right|}^2} = {\bf{\bar v}}_0^H{{\bf{Q}}_k}{{{\bf{\bar v}}}_0} \ge 2{\rm{Re}}\left( {{\bf{\bar v}}_0^{\left( l \right)H}{{\bf{Q}}_k}{{{\bf{\bar v}}}_0}} \right) - {\bf{\bar v}}_0^{\left( l \right)H}{{\bf{Q}}_k}{\bf{\bar v}}_0^{\left( l \right)}},
\end{align}
\begin{align}\label{C46-2dev}
{{{\left| {h_{d,k}^H + {\bf{q}}_k^H{{\bf{v}}_1}} \right|}^2} = {\bf{\bar v}}_1^H{{\bf{Q}}_k}{{{\bf{\bar v}}}_1} \ge 2{\rm{Re}}\left( {{\bf{\bar v}}_1^{\left( l \right)H}{{\bf{Q}}_k}{{{\bf{\bar v}}}_1}} \right) - {\bf{\bar v}}_1^{\left( l \right)H}{{\bf{Q}}_k}{\bf{\bar v}}_1^{\left( l \right)}},
\end{align}

where ${{\bf{Q}}_k} = {{{\bf{\bar q}}}_k}{\bf{\bar q}}_k^H$. As such, problem \eqref{C45} can be transformed into the following tractable form
\begin{subequations}\label{C46}
\begin{align}
\label{C46-a}\mathop {\max }\limits_{{S_k},{{\bf{v}}_0},{{\bf{v}}_1}} \;\;&B\sum\limits_{k = 1}^K {{\tau _{1,k}}{{\log }_2}\left( {1 + \frac{{{p_k}{S_k}}}{{{\sigma ^2}}}} \right)}  + \sum\limits_{k = 1}^K {\frac{{T{f_k}}}{C}}\\
\label{C46-b}{\rm{s.t.}}\;\;\;\;&{\tau _{1,k}}{p_k} + T{\gamma _c}f_k^3 \le \eta {P_E}{\tau _0}\left( {2{\rm{Re}}\left( {{\bf{\bar v}}_0^{\left( l \right)H}{{\bf{Q}}_k}{{{\bf{\bar v}}}_0}} \right) - {\bf{\bar v}}_0^{\left( l \right)H}{{\bf{Q}}_k}{\bf{\bar v}}_0^{\left( l \right)}} \right),\;\;\forall k, \\
\label{C46-c}&{S_k} \le 2{\mathop{\rm Re}\nolimits} \left( {{\bf{\bar v}}_1^{\left( l \right)H}{{\bf{Q}}_k}{{{\bf{\bar v}}}_1}} \right) - {\bf{\bar v}}_1^{\left( l \right)H}{{\bf{Q}}_k}{\bf{\bar v}}_1^{\left( l \right)}{\rm{ }}\;\;\forall k, \\
\label{C46-d}&\eqref{IRS-b},~\eqref{IRS-c},
\end{align}
\end{subequations}
where \eqref{C46} is a convex problem and thus it can be optimally solved by standard solvers.
\subsubsection{AO Algorithm Proposed for Solving $\left( {{\mathop{\rm P}\nolimits} _{{\rm{TDMA}}}^{{\rm{case2}}}} \right)$ }
\begin{algorithm}[!h]\label{method1}
\caption{AO Algorithm}
\begin{algorithmic}[1]
\STATE {\bf Initialization}: {Set iteration index to $l = 1$ and IRS BF vectors to ${{{\bf{\bar v}}}_0} \!\!=\!\! {\bf{\bar v}}_0^{\left( 1 \right)},{{{\bf{\bar v}}}_1} \!\!=\!\! {\bf{\bar v}}_1^{\left( 1 \right)}$}.
\STATE {\bf repeat}
\STATE Solve problem \eqref{C15} for given $\left\{ {{\bf{\bar v}}_0^{\left( l \right)},{\bf{\bar v}}_1^{\left( l \right)}} \right\}$ based on the interior point method and denote the optimal solution as ${\Xi ^{\left( l \right)}} = \left\{ {\tau _0^{\left( l \right)},\tau _{1,k}^{\left( l \right)},p_k^{\left( l \right)},f_k^{\left( l \right)}} \right\}$.
\STATE Solve problem \eqref{C46} for given ${\Xi ^{\left( l \right)}}$ and $\left\{ {{\bf{\bar v}}_0^{\left( l \right)},{\bf{\bar v}}_1^{\left( l \right)}} \right\}$, where the optimal solution is denoted by $\left\{ {{\bf{\bar v}}_0^{\left( {l + 1} \right)},{\bf{\bar v}}_1^{\left( {l + 1} \right)}} \right\}$.
\STATE Update $l=l+1$.
\STATE {\bf until} the fractional increase of the objective value falls below a threshold $\xi  > 0$.
\STATE Let ${\left[ {{\bf{\bar v}}_0^{l}} \right]_n} = {{{{\left[ {{\bf{\bar v}}_0^{l}} \right]}_n}} \mathord{\left/
 {\vphantom {{{{\left[ {{\bf{\bar v}}_0^{l}} \right]}_n}} {\left| {{{\left[ {{\bf{\bar v}}_0^{l}} \right]}_n}} \right|}}} \right.
 \kern-\nulldelimiterspace} {\left| {{{\left[ {{\bf{\bar v}}_0^{l}} \right]}_n}} \right|}},{\left[ {{\bf{\bar v}}_1^{l}} \right]_n} = {{{{\left[ {{\bf{\bar v}}_1^{l}} \right]}_n}} \mathord{\left/
 {\vphantom {{{{\left[ {{\bf{\bar v}}_1^{l}} \right]}_n}} {\left| {{{\left[ {{\bf{\bar v}}_1^{l}} \right]}_n}} \right|}}} \right.
 \kern-\nulldelimiterspace} {\left| {{{\left[ {{\bf{\bar v}}_1^{l}} \right]}_n}} \right|}}$, and solve problem \eqref{C15} for given $\left\{ {{\bf{\bar v}}_0^{\left( l \right)},{\bf{\bar v}}_1^{\left( l \right)}} \right\}$.
\end{algorithmic}
\end{algorithm}
An efficient AO algorithm where the IRS BF vectors and resource allocation are alternately optimized until convergence is achieved can be proposed. For arbitrary continuous amplitudes, the objective value of $({\mathop{\rm P}\nolimits} _{{\rm{TDMA}}}^{{\rm{case2}}})$ is non-decreasing by alternately optimizing $\left\{ {{\tau _0},{\tau _{1,k}},{p_k},{f_k}} \right\}$ and $\left\{ {{{{\bf{\bar v}}}_0},{{{\bf{\bar v}}}_1}} \right\}$, and also upper-bounded by a finite value. The proposed AO algorithm is guaranteed to converge when relaxing the unit-modulus constraints. The converged solution, denoted by $\left\{ {{\bf{\bar v}}_0^*,{\bf{\bar v}}_1^*} \right\}$ may not satisfy the unit-modulus constraints. In this case, the feasible IRS BF vectors, denoted by $\left\{ {{\bf{\bar v}}_0^\star,{\bf{\bar v}}_1^\star} \right\}$, can be obtained as ${\left[ {{\bf{\bar v}}_0^\star} \right]_n} = {{{{\left[ {{\bf{\bar v}}_0^*} \right]}_n}} \mathord{\left/
 {\vphantom {{{{\left[ {{\bf{\bar v}}_0^*} \right]}_n}} {\left| {{{\left[ {{\bf{\bar v}}_0^*} \right]}_n}} \right|}}} \right.
 \kern-\nulldelimiterspace} {\left| {{{\left[ {{\bf{\bar v}}_0^*} \right]}_n}} \right|}},{\left[ {{\bf{\bar v}}_1^\star} \right]_n} = {{{{\left[ {{\bf{\bar v}}_1^*} \right]}_n}} \mathord{\left/
 {\vphantom {{{{\left[ {{\bf{\bar v}}_1^*} \right]}_n}} {\left| {{{\left[ {{\bf{\bar v}}_1^*} \right]}_n}} \right|}}} \right.
 \kern-\nulldelimiterspace} {\left| {{{\left[ {{\bf{\bar v}}_1^*} \right]}_n}} \right|}},\forall n$. Finally, problem \eqref{C15} is solved for a given $\left\{ {{{{\bf{\bar v}}}_0^\star},{{{\bf{\bar v}}}_1^\star}} \right\}$ pair and the feasible solution obtained for $\left( {{\mathop{\rm P}\nolimits} _{{\rm{TDMA}}}^{{\rm{case2}}}} \right)$ is denoted by $\left\{ {{{{\bf{\bar v}}}_0^\star},{{{\bf{\bar v}}}_1^\star},{\tau _0^\star},{\tau _{1,k}^\star},{p_k^\star},{f_k^\star}} \right\}$. The details of the AO procedure are summarized in Algorithm 1.

It is worth noting that Algorithm 1 can be directly applied for obtaining the solution of $\left( {{\mathop{\rm P}\nolimits} _{{\rm{TDMA}}}^{{\rm{case1}}}} \right)$  by setting ${{\bf{v}}_1} = {{\bf{v}}_0}$. The details are omitted here for its simplicity.
\subsubsection{Complexity Analysis}
The computational complexity of Algorithm 1 is dominated by Step 3 and 4. Specifically, in Step 3, problem \eqref{C15} can be solved by the interior-point method, whose complexity is $O\left( {{{\left( {3K + 1} \right)}^{3.5}}} \right)$ \cite{convex}. In Step 4, the complexity of solving problem \eqref{C46} is ${\cal O}\left( {{{\left( {2N} \right)}^{3.5}}} \right)$. Therefore, the total complexity of Algorithm 1 is ${\cal O}\left( {{L_{{\rm{iter}}}}\left( {{{\left( {2N} \right)}^{3.5}} + {{\left( {3K + 1} \right)}^{3.5}}} \right)} \right)$, where ${L_{{\rm{iter}}}}$ denotes the number of iterations for Algorithm 1.
\vspace{-8pt}
\subsection{Extension to Problem $\left( {{\mathop{\rm P}\nolimits} _{{\rm{TDMA}}}^{{\rm{case3}}}} \right)$ }
In this subsection, we focus our attention on solving problem $\left( {{\mathop{\rm P}\nolimits} _{{\rm{TDMA}}}^{{\rm{case3}}}} \right)$ for the scenario, where different IRS BF vectors can be adopted in different TSs where each device offloads its own tasks. It may be readily shown that the $n$-th phase shift of the optimal ${\bf{v}}_{1,k}^* = {\left[ {{e^{j\varphi _{1,k}^*}}, \ldots ,{e^{j\varphi _{N,k}^*}}} \right]^T}$ is given by
\begin{align}\label{C48dev}
\varphi _{n,k}^* = \arg \left( {h_{d,k}^H} \right) + \arg \left( {{{\left[ {{{\bf{q}}_k}} \right]}_n}} \right),n = 1, \ldots N.
\end{align}
After determining the optimal ${\bf{v}}_{1,k}^* $, we use $g_k^{{\rm{off}}} = {\left| {h_{d,k}^H + {\bf{q}}_k^H{\bf{v}}_1^{\left( k \right)*}} \right|^2}$ to represent the channel power gain of device $k$ for the UL offloading link. Then, $\left( {{\mathop{\rm P}\nolimits} _{{\rm{TDMA}}}^{{\rm{case3}}}} \right)$ is rewritten as follows
\begin{subequations}\label{C49}
\begin{align}
\label{C49-a}\mathop {\max }\limits_{{\tau _0},{\tau _{1,k}},{p_k},{{\bf{v}}_0}{f_k}} \;\;&B\sum\limits_{k = 1}^K {{\tau _{1,k}}{{\log }_2}\left( {1 + \frac{{{p_k}g_k^{{\rm{off}}}}}{{{\sigma ^2}}}} \right)}  + \sum\limits_{k = 1}^K {\frac{{T{f_k}}}{C}}\\
\label{C49-b}{\rm{s.t.}}\;\;&~\eqref{C9-b}, ~\eqref{C9-c}, ~\eqref{C9-d}, ~\eqref{C9-e}.
\end{align}
\end{subequations}
Note that Algorithm 1 proposed in Section V-A can be directly applied to solving problem \eqref{C49} with slight modifications. Specifically, for a fixed $\left\{ {{{\bf{v}}_0}} \right\}$, the subproblem is simplified to problem \eqref{C15}, whose optimal solution can be efficiently solved by standard solvers. For the fixed $\left\{ {{\tau _0},{\tau _{1k}},{p_k},{f_k}} \right\}$, the subproblem with respect to $\left\{ {{{\bf{v}}_0}} \right\}$ can be efficiently solved by exploiting the SCA technique proposed in Section V-A. The details are omitted here for brevity.
\begin{rem}
Although the individual rate constraint of each device is not considered in this paper, the proposed algorithm is applicable to the corresponding problems subject to such constraints. Specifically, for \textbf{Case 1}, the individual rate constraint of each device is given by
\begin{align}\label{QoS}
B{\tau _{1,k}}{\log _2}\left( {1 + {p_k}\frac{{{{\left| {h_{d,k}^H + {\bf{q}}_k^H{{\bf{v}}_0}} \right|}^2}}}{{{\sigma ^2}}}} \right) + \frac{{T{f_k}}}{C} \ge {R_{k,\min }},\;\;\forall k,
\end{align}
where ${R_{k,\min }}$ is the minimum number of computational bits required by device $k$. For any given ${{{\bf{v}}_0}}$, by letting ${e_k} = {p_k}{\tau _{1,k}}$, constraint \eqref{QoS} can be transformed to
\begin{align}\label{QoS1}
B{\tau _{1,k}}{\log _2}\left( {1 + \frac{{{e_k}{{\left| {h_{d,k}^H + {\bf{q}}_k^H{{\bf{v}}_0}} \right|}^2}}}{{{\tau _{1,k}}{\sigma ^2}}}} \right) + \frac{{T{f_k}}}{C} \ge {R_{k,\min }},\;\;\forall k.
\end{align}
The left-hand-side of \eqref{QoS1} is concave with respect to $\left\{ {{e_k},{\tau _{1,k}},{f_k}} \right\}$. As such, \eqref{QoS1} is convex and thus it can be handled by standard solvers. For a given $\left\{ {{p_k},{\tau _{1,k}},{f_k}} \right\}$, \eqref{QoS} can be rewritten as
\begin{align}\label{QoS2}
{\left| {h_{d,k}^H + {\bf{q}}_k^H{{\bf{v}}_0}} \right|^2} \ge \frac{{{2^{{{\left( {R_k^{\min } - {{T{f_k}} \mathord{\left/
 {\vphantom {{T{f_k}} C}} \right.
 \kern-\nulldelimiterspace} C}} \right)} \mathord{\left/
 {\vphantom {{\left( {R_k^{\min } - {{T{f_k}} \mathord{\left/
 {\vphantom {{T{f_k}} C}} \right.
 \kern-\nulldelimiterspace} C}} \right)} {\left( {B{\tau _{1,k}}} \right)}}} \right.
 \kern-\nulldelimiterspace} {\left( {B{\tau _{1,k}}} \right)}}}} - 1}}{{{p_k}}},\;\;\forall k.
\end{align}
The left-hand-side of \eqref{QoS2} is convex with respect to ${{{\bf{v}}_0}}$ and thus it can be handled by the SCA technique proposed in Section V-A. Therefore, the proposed AO algorithm can be readily extended to solve the corresponding problem for \textbf{Case 1} subject to the individual rate constraint of each device.
Replacing ${{\bf{v}}_0}$ by ${{\bf{v}}_1}$, the method is also applicable to solving the associated problem for \textbf{Case 2}. For \textbf{Case 3}, the optimal ${{\bf{v}}_{1,k}}$ is obtained in \eqref{C48dev} and the resource allocation OP with respect to $\left\{ {{p_k},{\tau _{1,k}},{f_k},{\tau _0}} \right\}$ can be handled in a similar way.
\end{rem}

\vspace{-10pt}
\section{Numerical Results}
In this section, numerical results are provided for characterizing the performance of the proposed schemes and for gaining insights into the design and implementation of IRS-aided WP-MEC systems. The HAP and IRS are placed at $\left( {0,0,0} \right)$ m and $\left( {10,0,3} \right)$ m, respectively. The pathloss exponents of both the HAP-IRS and IRS-device channels are set to $2.2$, while those of the HAP-device channels are set to $3$. The signal attenuation at a reference distance of 1 m is set as 30 dB. Furthermore, the bandwidth and noise power are set to $B = 500$ kHz and ${\sigma ^2} =  - 75$ dBm, respectively. The system parameters related to the EH and computational model are set as follows: $\eta =0.8$, ${\gamma _c} = {10^{ - 28}}$, $C = \left[ {400,800,2000} \right]$ cycles/bit and $T = 1\;{\rm{s}}$.
\vspace{-10pt}
\subsection{Activation Condition in Single-User Setup}
To verify our analysis of the UL offloading activation condition in the single-user setup, we provide the following pair of numerical examples. The location of a typical device is set to $\left( {12,0,0} \right)$ m.
\begin{figure*}[t!]
\centering
\subfigure[$C = 400$ cycles/bit.]{\label{activation_condition1}
\includegraphics[width= 2.5in, height=2in]{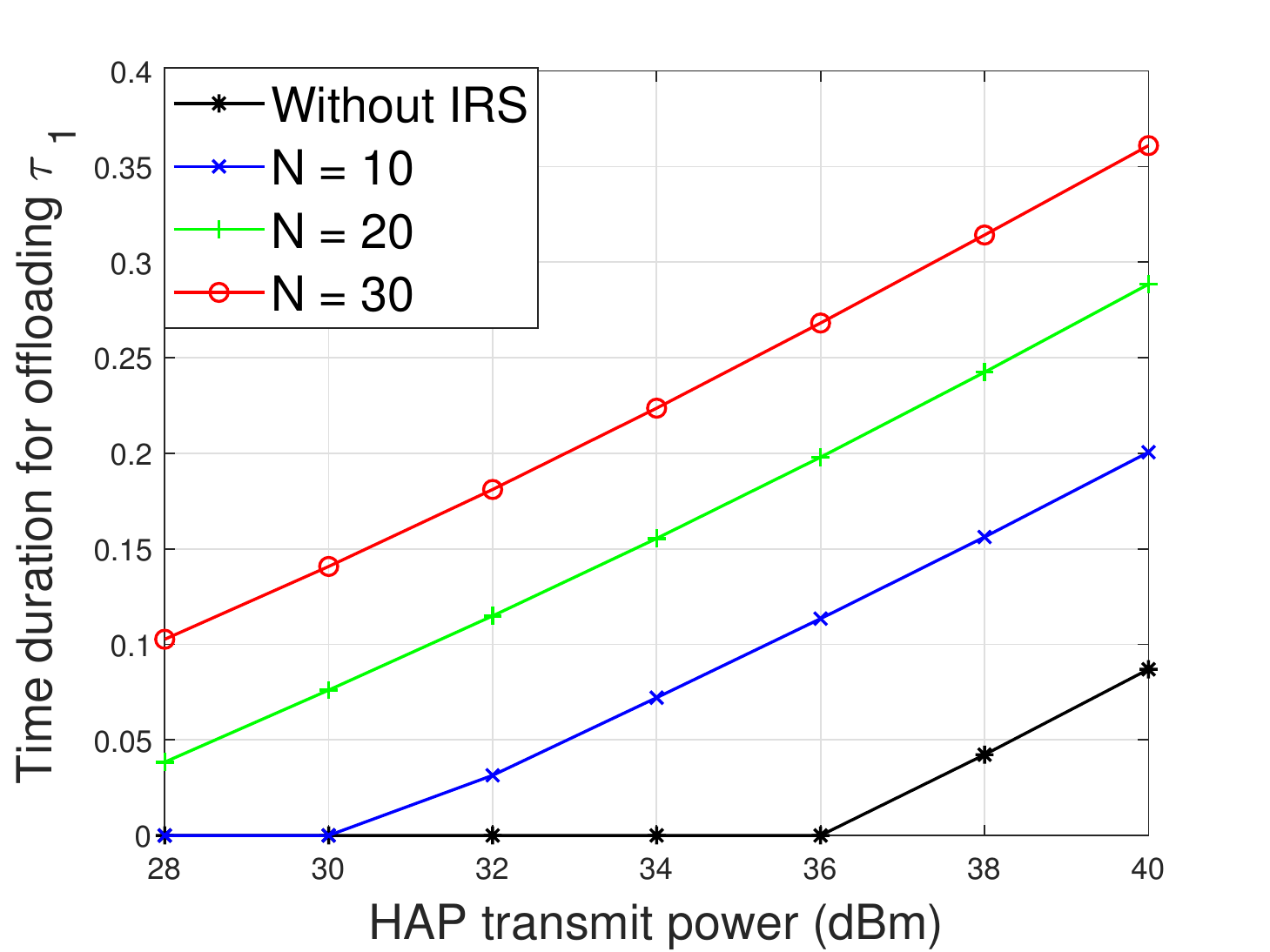}}
\subfigure[$C = 800$ cycles/bit.]{\label{activation_condition2}
\includegraphics[width= 2.5in, height=2in]{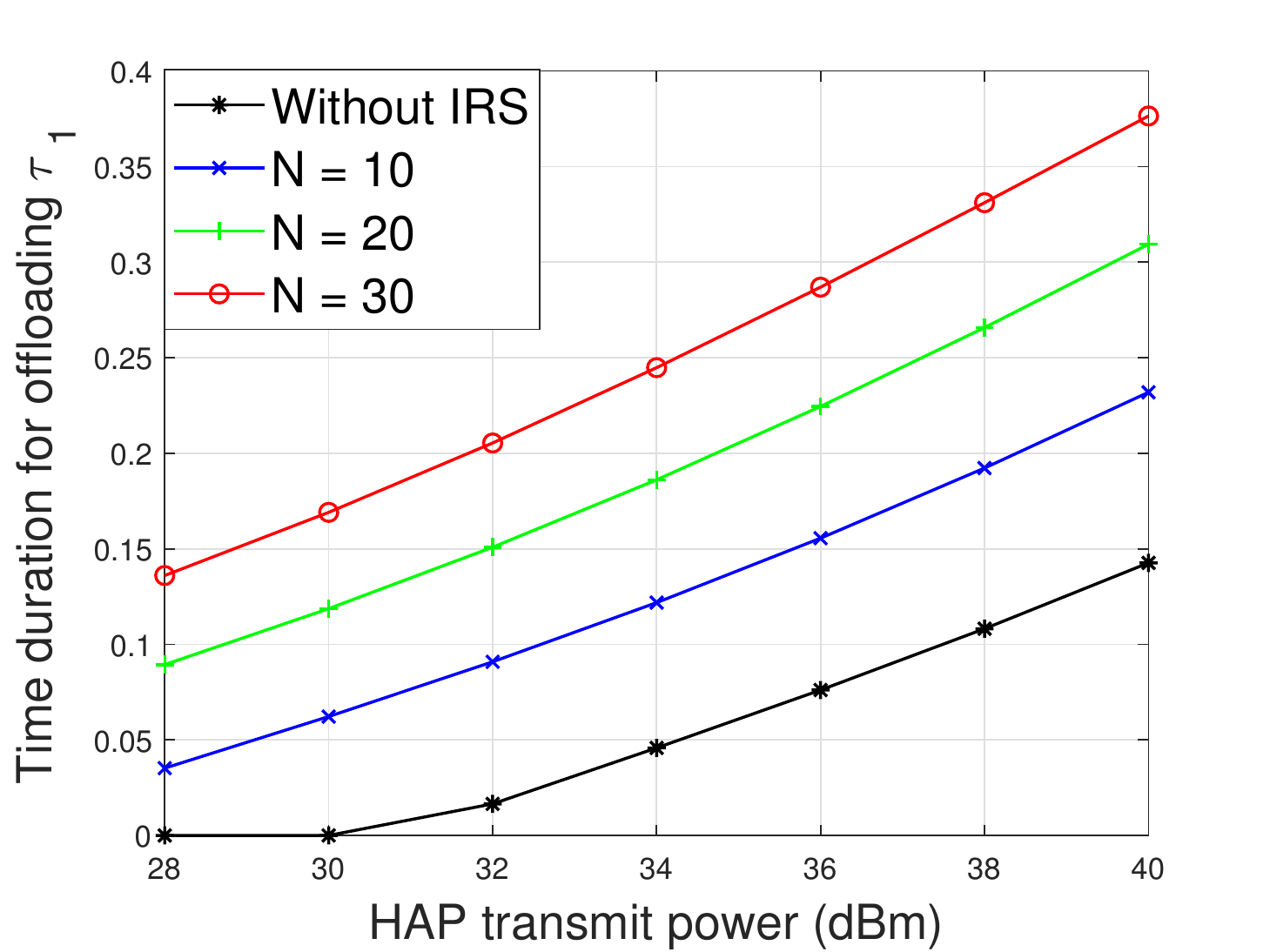}}
\setlength{\abovecaptionskip}{0.4cm}
\caption{{Illustration of UL offloading condition with different configurations.}}\label{activation contition}
\vspace{-16pt}
\end{figure*}
For the illustration of the optimal ${\tau _1}$, we assume that all the links are LoS dominated, i.e., the Rician factor is high. The time allocated to UL offloading under $C = \left[ {400,800} \right]$ cycles/bit is presented in Fig. \ref{activation contition}. The region where ${\tau _1} > 0$ indicates that the UL offloading is activated. It can be observed that the transmit power ${P_E}$ of the HAP which ensures ${\tau _1} > 0$ decreases as $N$ increases. This is consistent with Proposition 2, because the device is more likely to offload, when enjoying better channel conditions and the channel power gain can be improved by increasing the number of IRS elements $N$. Additionally, we can observe that the device tends to offload, when $C$ becomes higher. This is because the computation rate attained by local computing becomes marginal at large $C$, which forces the device to offload more tasks for improving the computation rate. This is also consistent with \eqref{C30dev}, namely that the threshold used for activating UL offloading decreases as $C$ increases.
\vspace{-10pt}
\subsection{Performance Comparison}
Next, we consider a general multi-user setup to provide further performance comparisons and for demonstrating the efficiency of the proposed solutions. Specifically, five devices are uniformly and randomly distributed within a radius of 1.5 m centered at $\left( {10,0,0} \right)$ m. The small scale fading of all links is characterized by a Rician factor of $2$. We set ${P_E} = 40$ dBm in this subsection.

\subsubsection{Efficiency of IRSs in WP-MEC Systems}
To demonstrate the efficiency of IRSs in WP-MEC systems, the following benchmark schemes are considered for comparision: 1) Proposed AO in Algorithm 1 to solve $\left( {{\mathop{\rm P}\nolimits} _{{\rm{TDMA}}}^{{\rm{case1}}}} \right)$; 2) Fixed WPT time but optimizing all other variables; 3) Fixed IRS phase shifts but optimizing all other variables; 4) Without IRS. In Fig. \ref{RA}, we plot the average total number of computed bits versus the number of IRS elements. It is observed that the average total number of computed bits output by our proposed Algorithm 1 over the benchmark schemes increases upon increasing $N$. Additionally, the performance gain of the scheme using the fixed phase shifts over the system without IRS is marginal, which highlights the importance of carefully designing the IRS BF. Moreover, for small $N$, the scheme using fixed WPT duration performs even worse than that without IRS, but for large $N$, it significantly outperforms the fixed phase shifts-based scheme. This demonstrates that the gain of IRS BF compensates the performance loss due to a fixed WPT duration. Nevertheless, the results unveil the importance of the joint design of the WPT duration and IRS BF.

\begin{figure}[t!]
\centering
\includegraphics[width=3in]{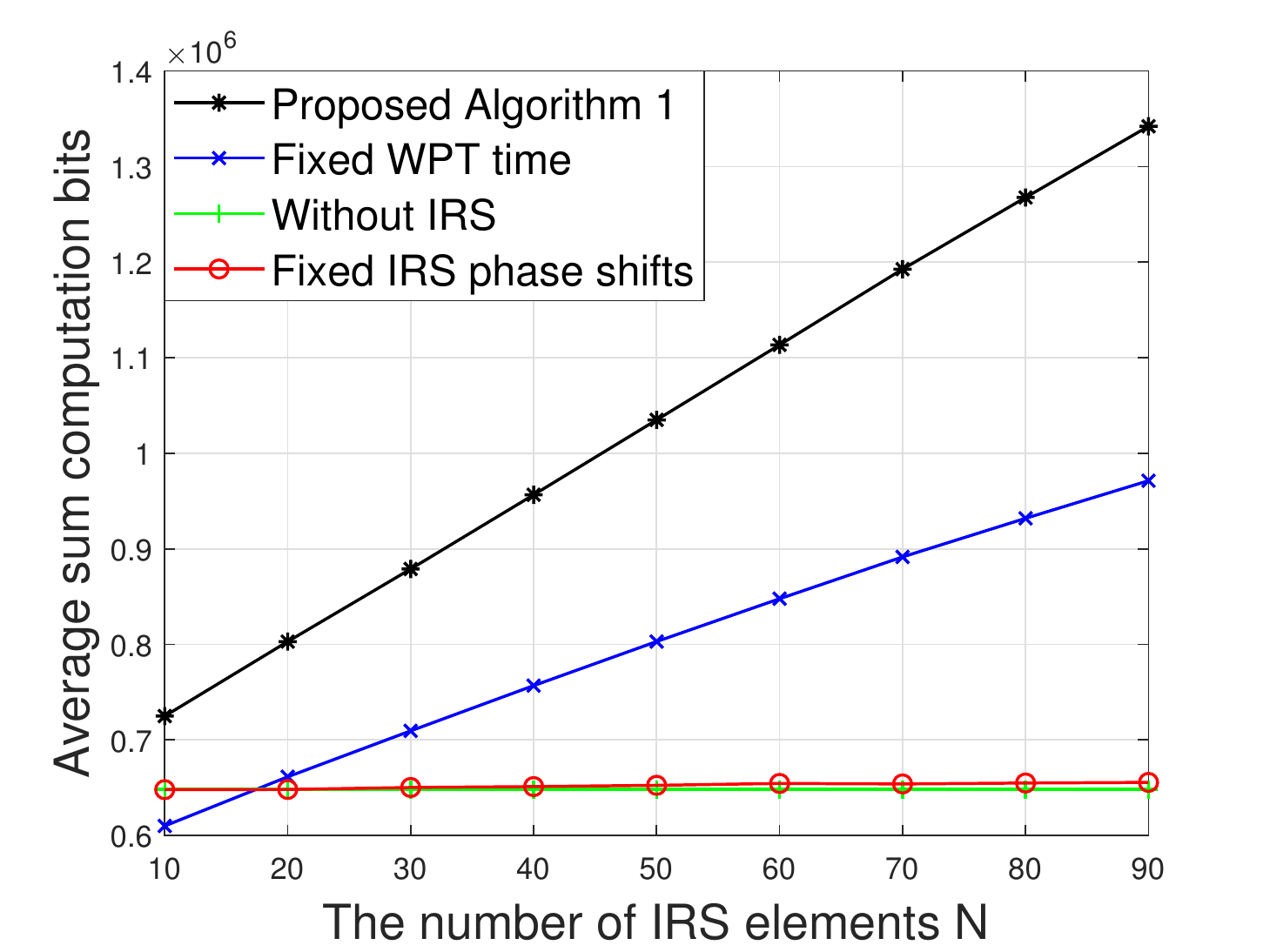}
\caption{{Performance Comparison with different resource allocation schemes when $C = 2000$ cycles/bit. }}
\label{RA}
\vspace{-16pt}
\end{figure}

\begin{figure*}[t!]
\centering
\subfigure[Impact of $N$ on WPT time.]{\label{WPT time}
\includegraphics[width= 2.5in, height=2.0in]{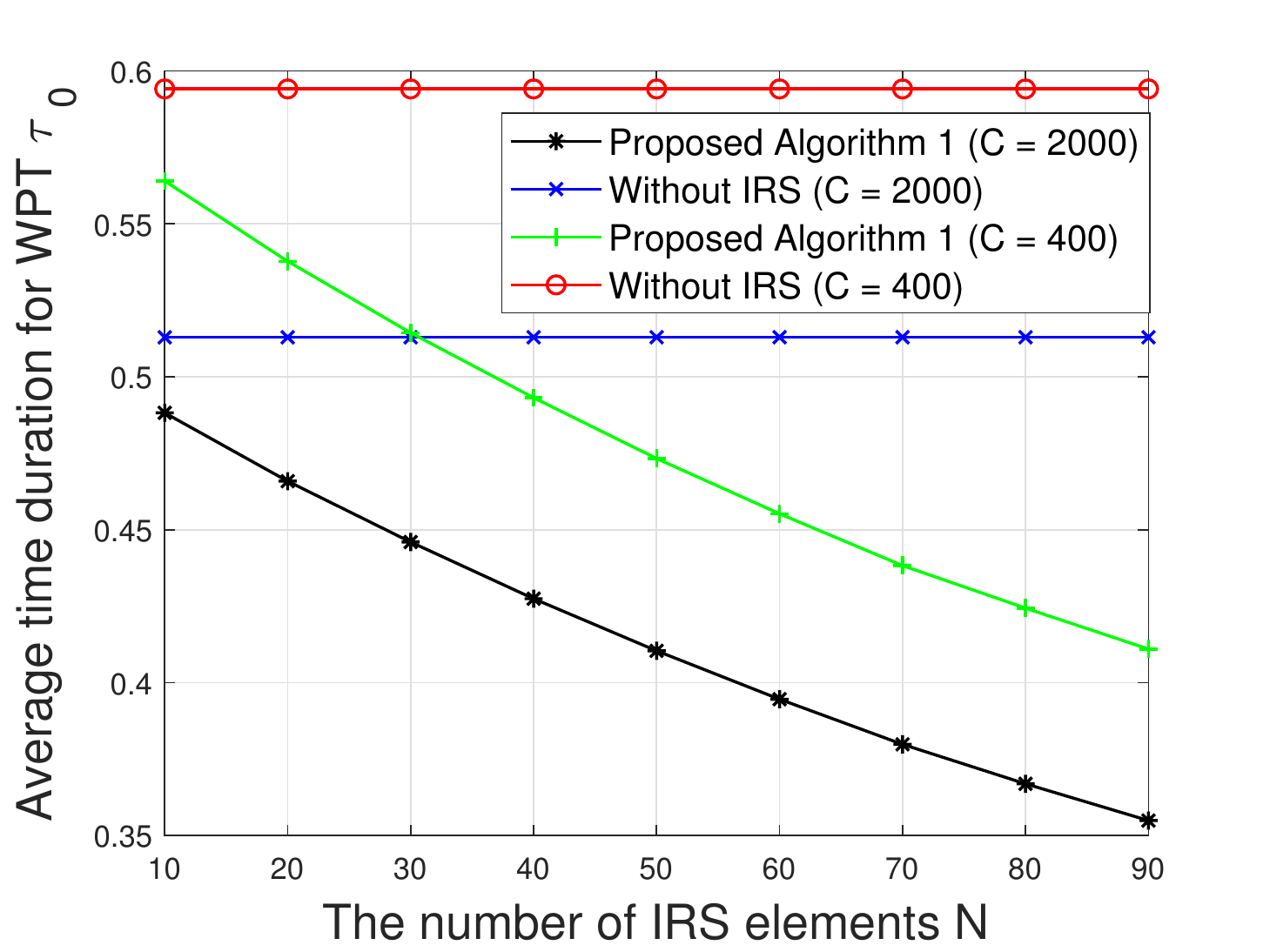}}
\subfigure[Impact of $N$ on the harvested energy.]{\label{Harvested energy}
\includegraphics[width= 2.5in, height=2.0in]{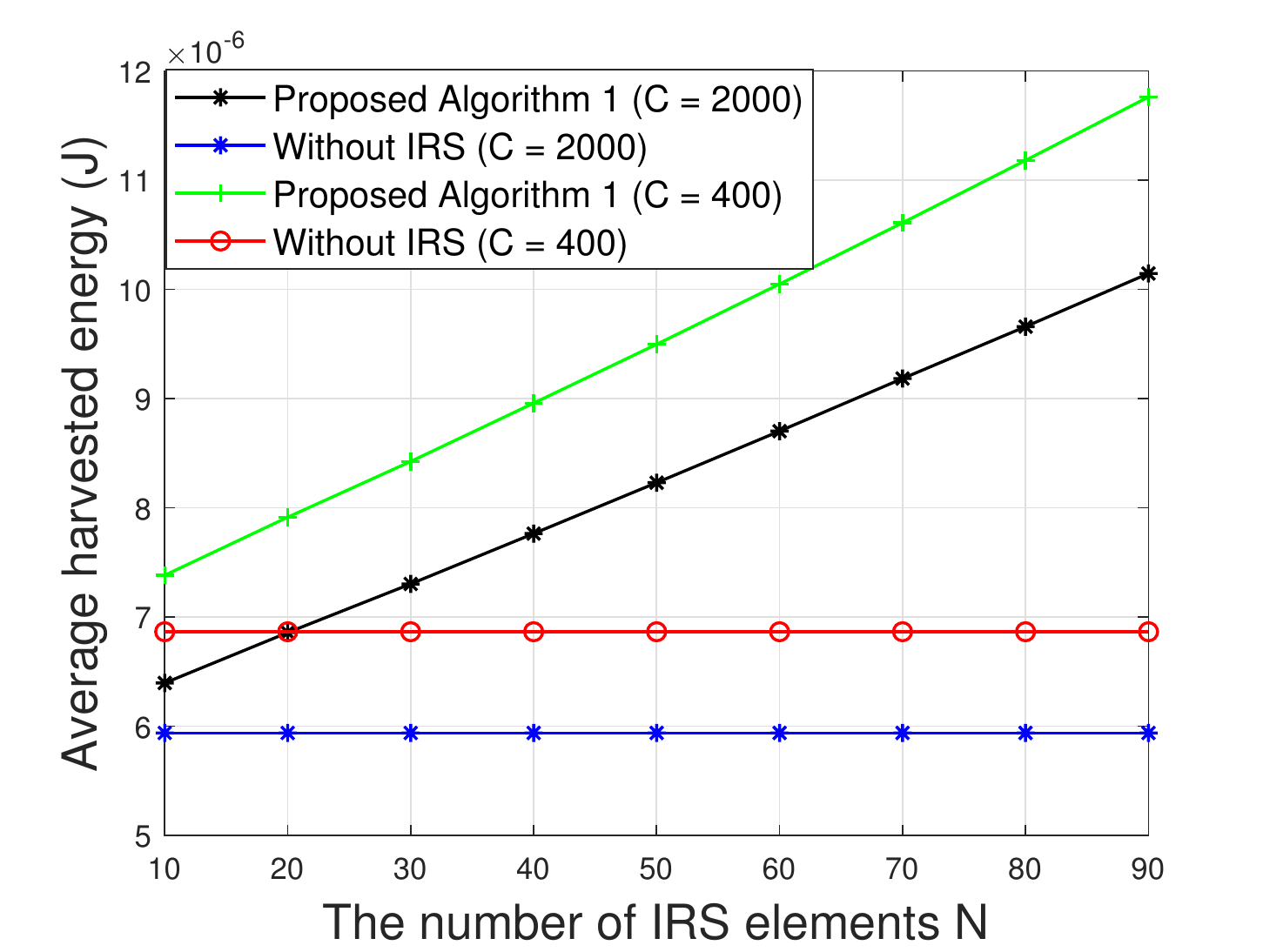}}
\setlength{\abovecaptionskip}{0.4cm}
\caption{{Impact of $N$ on WPT time and harvested energy.}}\label{time and energy}
\vspace{-16pt}
\end{figure*}

To further demonstrate the benefits brought out by IRSs for WP-MEC systems, we investigate the impact of $N$ both on the DL WPT duration and on the total energy harvested at each device. As shown in Fig. \ref{WPT time}, the optimized WPT duration decreases with $N$ for both $C = 400$ and $C = 2000$, which indicates that the energy consumed at the HAP, namely ${E_{{\rm{HAP}}}} = {\tau _0}{P_E}$, can be reduced by increasing $N$. Meanwhile, more time can be reserved for each device's UL offloading, which increases the total number of computed bits. This implies that embedding IRSs into WP-MEC systems achieves both computation rate improvements and energy consumption reductions. Although a higher $N$ leads to a reduced DL WPT time ${\tau _0}$, Fig. \ref{Harvested energy} shows that the total energy harvested by each device even increases with $N$. This is because the energy signal reflected by the IRS towards devices becomes more focused, which in turn improves the efficiency of WPT upon increasing $N$. Indeed, the high passive BF gain attained by IRSs increases the degrees of freedom for enhancing the flexibility of resource allocation design. Thanks to the improved channel conditions granted for both the DL WPT and UL offloading links, more time is available for offloading, while maintaining sufficient harvested energy, which achieves substantial computation rate improvements.

\subsubsection{Comparison of Different Computational Modes}
For comparison, we consider different computational modes as follows: 1) The partial offloading mode: Algorithm 1 is applied for solving $\left( {{\mathop{\rm P}\nolimits} _{{\rm{TDMA}}}^{{\rm{case1}}}} \right)$; 2) Offloading only: The algorithm proposed in \cite{9400380} is adopted for solving the resultant problem, when all devices only perform UL offloading; 3) Local computing only: All the devices only perform local computing and ${{\bf{v}}_0}$ is optimized based on the method in Section V-A. The average total number of computed bits versus $N$ is plotted in Fig. \ref{Computation modes} for different values of $C$. As expected, the partial offloading mode performs the best among all the specific schemes. The reason for this trend is that all the devices can flexibly select their computational mode based on the specific channel conditions under the partial offloading mode. Additionally, the offloading only mode significantly outperforms the local computing mode. This is because the IRS can only benefit the DL WPT for local computing mode, while the efficiency of both DL WPT and UL offloading can be improved with the aid of IRSs for the offloading-only mode. Moreover, it is observed that the gain of partial offloading mode over the offloading only-mode becomes marginal for large $C$, since the number of bits computed locally is low for a high $C$.
\begin{figure}
\begin{minipage}[t]{0.45\linewidth}
\centering
\includegraphics[width=2.9in]{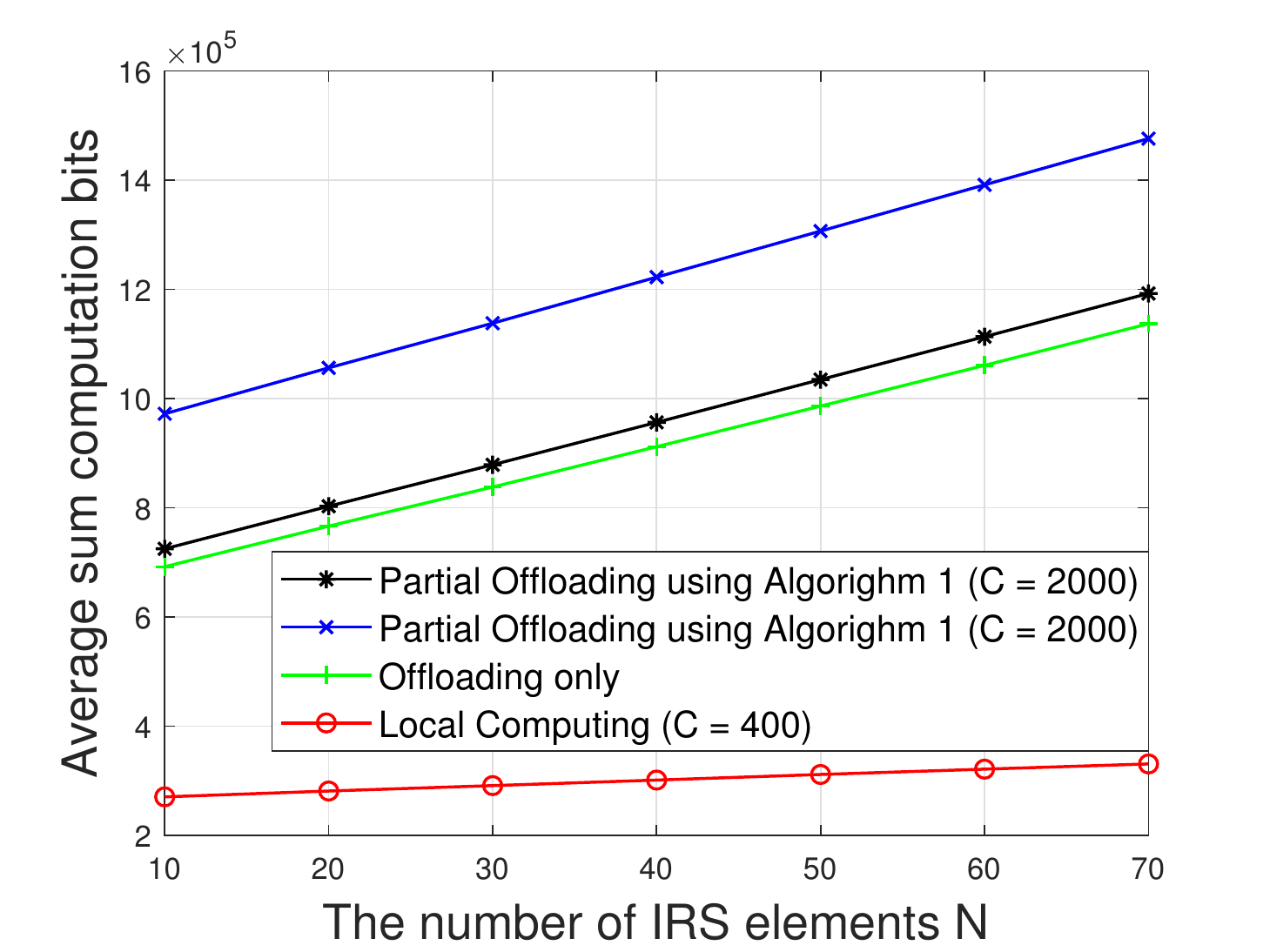}
\caption{Performance comparison of different computation modes.}
\label{Computation modes}
\end{minipage}%
\hfill
\begin{minipage}[t]{0.45\linewidth}
\centering
\includegraphics[width=2.9in]{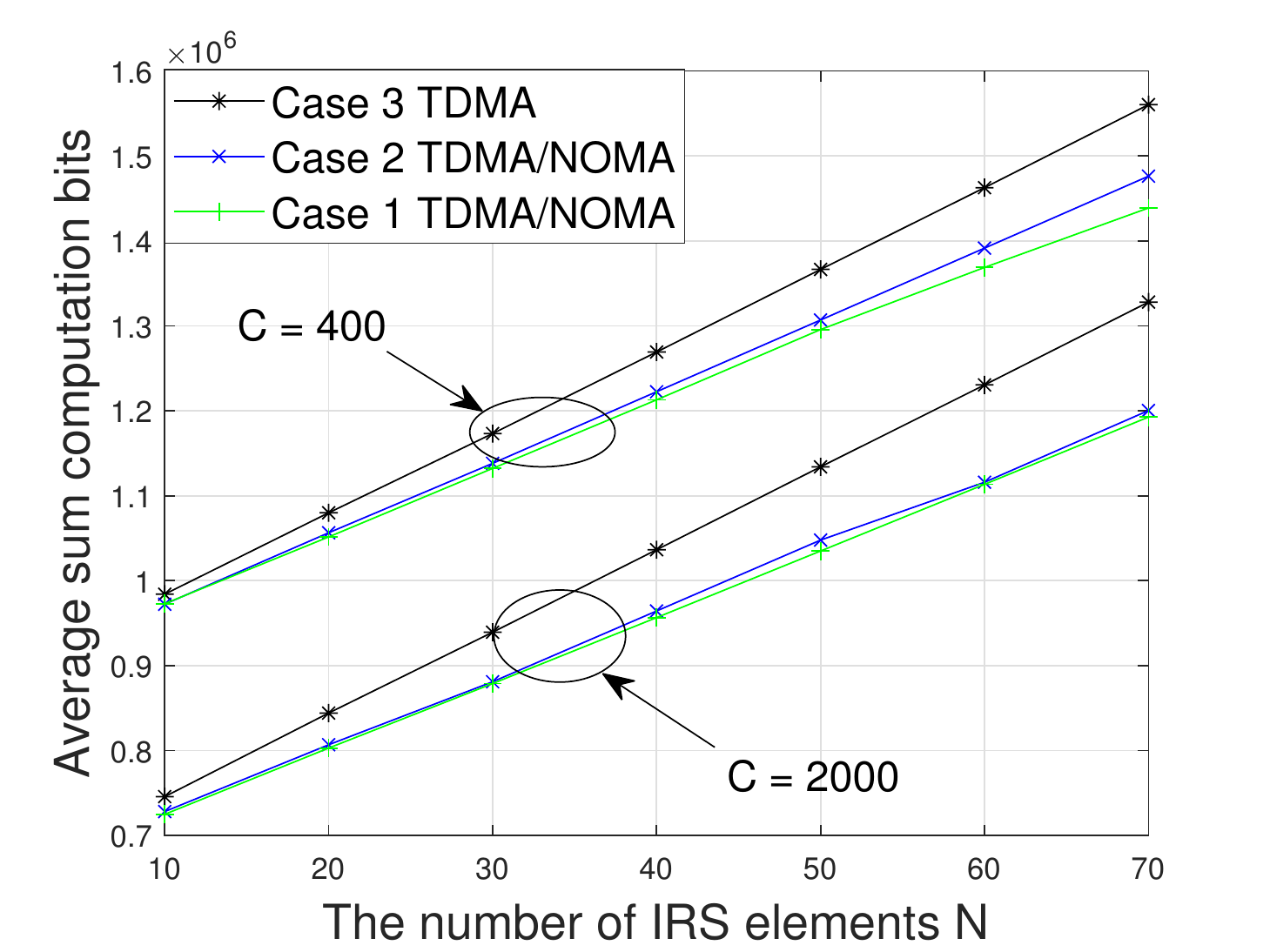}
\caption{Performance comparison of different DIBF schemes.}
\label{Dynamical BF}
\end{minipage}
\vspace{-16pt}
\end{figure}

\subsubsection{Impact of DIBF}
To illustrate the impact of DIBF on the computation rate, the average total number of computed bits versus $N$ is presented in Fig. \ref{Dynamical BF} for three cases. Note that the performance of NOMA is the same as that of TDMA for both \textbf{Case 1} and \textbf{Case 2}. It is observed in Fig. \ref{Dynamical BF} that the computation rate difference between \textbf{Case 3} and \textbf{Case 1}/\textbf{Case 2} expands as $N$ increases, which highlights the potential benefits of using dynamic IRS BF in TDMA-based UL offloading. The results reveal that the performance of TDMA may in fact become better than NOMA by using different IRS BF vectors for UL offloading in IRS-aided WP-MEC systems. Additionally, the performance gain of \textbf{Case 3} over \textbf{Case 1}/\textbf{Case 2} becomes more significant, when $C$ is high. This is because using dedicated IRS BF vectors for UL offloading only improves the computation rate contributed by UL offloading, but has no effect on local computing. Furthermore, for high $C$, the computation rate is dominated by that of UL offloading, while that of local computing is negligible. Finally, \textbf{Case 2} only attains a marginal gain over \textbf{Case 1}, especially for high $C$, which is in line with our analysis in Remark 2. The results suggest that using the same IRS BF vector is appropriate for both the DL WPT and UL offloading in scenarios, where the devices have weak computing capability and/or the system is sensitive to the signaling overhead.

\section{Conclusion}

The achievable computation rate performance of IRS-aided WP-MEC systems was studied in this paper. By taking into account the interplay between IRSs and the MA schemes, we answered a fundamental question: Does NOMA improve the achievable computation rate of IRS-aided WP-MEC systems compared to traditional TDMA? We first unveiled that offloading using a TDMA scheme achieves a better computation rate than that of NOMA, when the IRS BF vector can be flexibly adapted for UL offloading. The results provide important guidelines for selecting MA schemes for UL offloading in IRS-aided WP-MEC systems: it is preferable to use TDMA instead of NOMA for improving computation rate at the cost of extra signaling overhead. Furthermore, we proposed computationally efficient algorithms for maximizing computation rate under different DIBF schemes. Our numerical results validated the efficiency of our design over the benchmark schemes and also confirmed the benefits of IRSs in WP-MEC systems under different setups. It is our hope that this paper could provide guideline towards the application and optimization of specific IoE scenarios in next generation wireless networks
via reasonably selecting MA schemes.

\vspace{-8pt}
\begin{appendices}

\vspace{-10pt}
\section{Proof of Theorem 1}
\vspace{-6pt}
The proof starts by showing that $R_{{\rm{TDMA}}}^{{\rm{case2}}} \le R_{{\rm{NOMA}}}^{{\rm{case2}}}$. We denote the set characterizing the devices whose UL offloading is activated as ${{{\cal K}_{{\rm{off}}}}}$. Given that ${{\bf{v}}_0} = {\bf{v}}_0^*,{{\bf{v}}_1} = {\bf{v}}_1^*,{f_k} = f_k^*$ for $({\mathop{\rm P}\nolimits} _{{\rm{TDMA}}}^{{\rm{case2}}})$, the optimal ${p_k}$ can be expressed as $p_k^* = \left( {{\tau _0}\eta {P_E}{{\left| {h_{d,k}^H + {\bf{q}}_k^H{\bf{v}}_0^*} \right|}^2} - T{\gamma _c}f_k^{*3}} \right){\left( {{\tau _{1,k}}} \right)^{ - 1}}$ for ${k \in {{\cal K}_{{\rm{off}}}}}$, because each device will deplete all of its  energy. Now, we first discuss some properties of $\tau _{1,k}^*$ and $\tau _0^*$. To this end, $({\mathop{\rm P}\nolimits} _{{\rm{TDMA}}}^{{\rm{case2}}})$ can be simplified by optimizing ${{\tau _0}}$ and ${{\tau _{1,k}}}$ as follows:
\begin{subequations}\label{C51}
\begin{align}
\label{C51-a}{\rm{ }}\mathop {\max {\rm{ }}}\limits_{{{\tau _{1,k}}},{\tau _0}}\;\;&B\sum\limits_{k \in {{\cal K}_{{\rm{off}}}}} {{\tau _{1,k}}{{\log }_2}\left( {1 + \frac{{{\tau _0}\eta {P_E}{{\left| {h_{d,k}^H + {\bf{q}}_k^H{\bf{v}}_0^*} \right|}^2} - T{\gamma _c}f_k^{*3}}}{{{\tau _{1,k}}{\sigma ^2}}}{{\left| {h_{d,k}^H + {\bf{q}}_k^H{\bf{v}}_1^*} \right|}^2}} \right)}\\
\label{C51-b}{\rm{s.t.}}\;\;&\eqref{C9-c}.
\end{align}
\end{subequations}
Note that problem \eqref{C51} is a convex problem and its Lagrangian function is
\begin{align}\label{C52dev}
{{\cal L}_{\left( {{\rm{TDMA}}} \right)}}\left( {{\tau _0},{{\tau _{1,k}}},\lambda } \right) =& B\sum\limits_{k \in {{\cal K}_{{\rm{off}}}}} {{\tau _{1,k}}{{\log }_2}\left( {1 + \frac{{{\tau _0}\eta {P_E}{{\left| {h_{d,k}^H + {\bf{q}}_k^H{\bf{v}}_0^*} \right|}^2} - T{\gamma _c}f_k^{*3}}}{{{\tau _{1,k}}{\sigma ^2}}}{{\left| {h_{d,k}^H + {\bf{q}}_k^H{\bf{v}}_1^*} \right|}^2}} \right)}  \nonumber\\
& + \lambda \left( {T - {\tau _0} - \sum\limits_{k = 1}^K {{\tau _{1,k}}} } \right),
\end{align}
where $\lambda  \ge 0$ is the dual variable associated with \eqref{C51-b}. According to the Karush-Kuhn-Tucker (KKT) conditions, we have
\begin{align}\label{C53dev}
&\frac{{\partial {{\cal L}_{\left( {{\rm{TDMA}}} \right)}}\left( {{\tau _0},{{\tau _{1,k}}},\lambda } \right)}}{{\partial {{\tau _{1,k}}}}} = M\left( {{\Gamma _k}} \right) \buildrel \Delta \over = B\left( {{{\log }_2}\left( {1 + {\Gamma _k}} \right) - \frac{{{\Gamma _k}}}{{\left( {1 + {\Gamma _k}} \right)\ln 2}}} \right) - \lambda  = 0,\\
&\frac{{\partial {{\cal L}_{\left( {{\rm{TDMA}}} \right)}}\left( {{\tau _0},{{\tau _{1,k}}},\lambda } \right)}}{{\partial {\tau _0}}} = B\frac{{{\tau _0}\eta {P_E}{{\left| {h_{d,k}^H + {\bf{q}}_k^H{\bf{v}}_0^*} \right|}^2}{{\left| {h_{d,k}^H + {\bf{q}}_k^H{\bf{v}}_1^*} \right|}^2}}}{{{\sigma ^2}\left( {1 + {\Gamma _k}} \right)\ln 2}} - \lambda  = 0{\rm{ }},
\end{align}
where
\begin{align}\label{C55dev}
{\Gamma _k} = \frac{{{\tau _0}\eta {P_E}{{\left| {h_{d,k}^H + {\bf{q}}_k^H{\bf{v}}_0^*} \right|}^2} - T{\gamma _c}f_k^{*3}}}{{{\tau _{1,k}}{\sigma ^2}}}{\left| {h_{d,k}^H + {\bf{q}}_k^H{\bf{v}}_1^*} \right|^2}
\end{align}
denotes the received signal-to-noise ratio (SNR) of device $k$. Since $M\left( {{\Gamma _k}} \right)$ is an increasing function with respect to ${\Gamma _k}$, each device shares the same SNR at the optimal solution, i.e., ${\Gamma _k} = {\Gamma _m} = {\Gamma ^*},\forall k,m$. In particular, ${\Gamma ^*}$ is the solution of the equation:
\begin{align}\label{C56dev}
H\left( \Gamma  \right) \buildrel \Delta \over = {\log _2}\left( {1 + \Gamma } \right) - \frac{\Gamma }{{\left( {1 + \Gamma } \right)\ln 2}} - \frac{{\eta {P_E}{{\left| {h_{d,k}^H + {\bf{q}}_k^H{\bf{v}}_0^*} \right|}^2}{{\left| {h_{d,k}^H + {\bf{q}}_k^H{\bf{v}}_1^*} \right|}^2}}}{{{\sigma ^2}\left( {1 + \Gamma } \right)\ln 2}} = 0,
\end{align}
which can be readily obtained by applying the bisection search method, since $H\left( \Gamma  \right)$ is an increasing function with respect to $\Gamma$. Accordingly, $\tau _{1,k}^*$ and $\tau _0^*$ are given by
\begin{align}\label{C57dev}
\tau _0^* = \frac{{T + \sum\nolimits_{k \in {{\cal K}_{{\rm{off}}}}} {\frac{{T{\gamma _c}f_k^{*3}{{\left| {h_{d,k}^H + {\bf{q}}_k^H{\bf{v}}_1^*} \right|}^2}}}{{{\Gamma ^*}{\sigma ^2}}}} }}{{1 + \sum\nolimits_{k \in {{\cal K}_{{\rm{off}}}}} {\frac{{\eta {P_E}{{\left| {h_{d,k}^H + {\bf{q}}_k^H{\bf{v}}_0^*} \right|}^2}{{\left| {h_{d,k}^H + {\bf{q}}_k^H{\bf{v}}_1^*} \right|}^2}}}{{{\Gamma ^*}{\sigma ^2}}}} }},\tau _{1,k}^*{\rm{ = }}\frac{{\tau _0^*\eta {P_E}{{\left| {h_{d,k}^H + {\bf{q}}_k^H{\bf{v}}_0^*} \right|}^2} - T{\gamma _c}f_k^{*3}}}{{{\Gamma ^*}{\sigma ^2}{{\left| {h_{d,k}^H + {\bf{q}}_k^H{\bf{v}}_1^*} \right|}^{ - 2}}}},
\end{align}
respectively. Since ${\Gamma _k}{\rm{ = }}{\Gamma _m} = {\Gamma ^*},\forall k,m$, the optimal value of $({\mathop{\rm P}\nolimits} _{{\rm{TDMA}}}^{{\rm{case2}}})$ can be rewritten as
\begin{align}\label{C58dev}
R_{{\rm{TDMA}}}^{{\rm{case2}}} \!\!=\!\! B\tau _1^*{\log _2}\left( {1 \!\!+\!\! \frac{{\sum\nolimits_{{k \in {{\cal K}_{{\rm{off}}}}}} {\left( {\tau _0^*\eta {P_E}{{\left| {h_{d,k}^H + {\bf{q}}_k^H{\bf{v}}_0^*} \right|}^2} - T{\gamma _c}f_k^{*3}} \right){{\left| {h_{d,k}^H + {\bf{q}}_k^H{\bf{v}}_1^*} \right|}^2}} }}{{\tau _1^*{\sigma ^2}}}} \right)\!\!+ \!\!\sum\limits_{k = 1}^K {\frac{{Tf_k^*}}{C}},
\end{align}
where $\tau _1^* = \sum\nolimits_{k \in {{\cal K}_{{\rm{off}}}}} {\tau _{1,k}^*}$. It is plausible that $\left\{ {\tau _0^*,\tau _1^*,p_k^*,{\bf{v}}_0^*,{\bf{v}}_1^*,f_k^*} \right\}$ is a feasible solution for  $({\mathop{\rm P}\nolimits} _{{\rm{NOMA}}}^{{\rm{case2}}})$, which yields $R_{{\rm{TDMA}}}^{{\rm{case2}}} \le R_{{\rm{NOMA}}}^{{\rm{case2}}}$.

Next, we show that $R_{{\rm{TDMA}}}^{{\rm{case2}}} \ge R_{{\rm{NOMA}}}^{{\rm{case2}}}$. At the optimal solution of $({\mathop{\rm P}\nolimits} _{{\rm{NOMA}}}^{{\rm{case2}}})$, we can always construct a new solution satisfying ${{\tilde \tau }_0} = \tau _0^ \star ,\sum\nolimits_{{k \in {{\cal K}_{{\rm{off}}}}}} {{{\tilde \tau }_{1,k}}}  = \tau _1^ \star $, so that all devices share the same received SNR in the TDMA case, i.e.,
\begin{align}\label{C59dev}
\frac{{\left( {{{\tilde \tau }_0}\eta {P_E}{{\left| {h_{d,k}^H + {\bf{q}}_k^H{\bf{v}}_0^ \star } \right|}^2} - T{\gamma _c}{{\left( {f_k^ \star } \right)}^3}} \right)}}{{{{\tilde \tau }_{1,k}}{\sigma ^2}{{\left| {h_{d,k}^H + {\bf{q}}_k^H{\bf{v}}_1^ \star } \right|}^{ - 2}}}} = \frac{{\left( {{{\tilde \tau }_0}\eta {P_E}{{\left| {h_{d,m}^H + {\bf{q}}_m^H{\bf{v}}_0^ \star } \right|}^2} - T{\gamma _c}{{\left( {f_m^ \star } \right)}^3}} \right)}}{{{{\tilde \tau }_{1,m}}{\sigma ^2}{{\left| {h_{d,m}^H + {\bf{q}}_m^H{\bf{v}}_1^ \star } \right|}^{ - 2}}}},\forall k \ne m.
\end{align}
It is easy to verify that $\left\{ {{{\tilde \tau }_0},{{\tilde \tau }_{1,k}},\forall k} \right\}$ is also feasible for $({\mathop{\rm P}\nolimits} _{{\rm{TDMA}}}^{{\rm{case2}}})$ and always exists, which yields
\begin{align}\label{C60dev}
\tilde R_{{\rm{TDMA}}}^{{\rm{case2}}}&= B\sum\limits_{{k \in {{\cal K}_{{\rm{off}}}}}} {{{\tilde \tau }_{1,k}}} {\log _2}\left( {1 + \frac{{{{\tilde \tau }_0}\eta {P_E}{{\left| {h_{d,m}^H + {\bf{q}}_m^H{\bf{v}}_0^ \star } \right|}^2} - T{\gamma _c}{{\left( {f_k^ \star } \right)}^3}}}{{{{\tilde \tau }_{1,k}}{\sigma ^2}}}{{\left| {h_{d,m}^H + {\bf{q}}_m^H{\bf{v}}_1^ \star } \right|}^2}} \right) + \sum\limits_{k = 1}^K {\frac{{Tf_k^ \star }}{C}} \nonumber\\
&\mathop  = \limits^{\left( a \right)}   B\tau _1^ \star {\log _2}\left( {1 + \frac{{\sum\nolimits_{{k \in {{\cal K}_{{\rm{off}}}}}} {\left( {{{\tilde \tau }_0}{{\left| {h_{d,m}^H + {\bf{q}}_m^H{\bf{v}}_0^ \star } \right|}^2} - T{\gamma _c}{{\left( {f_k^ \star } \right)}^3}} \right){{\left| {h_{d,m}^H + {\bf{q}}_m^H{\bf{v}}_1^ \star } \right|}^2}} }}{{\tau _1^ \star {\sigma ^2}}}} \right) + \sum\limits_{k = 1}^K {\frac{{Tf_k^ \star }}{C}} \nonumber\\
&{\rm{           = }}R_{{\rm{NOMA}}}^{{\rm{case2}}},
\end{align}
where (a) follows that $\sum\nolimits_{{k \in {{\cal K}_{{\rm{off}}}}}} {{{\tilde \tau }_{1,k}}}  = \tau _1^ \star $ and all devices share the same SNR at the constructed solution. At the optimal solution of $({\mathop{\rm P}\nolimits} _{{\rm{TDMA}}}^{{\rm{case2}}})$, it follows that $R_{{\rm{TDMA}}}^{{\rm{case2}}} \ge R_{{\rm{NOMA}}}^{{\rm{case2}}}$.

Given $R_{{\rm{TDMA}}}^{{\rm{case2}}} \le R_{{\rm{NOMA}}}^{{\rm{case2}}}$ and $R_{{\rm{TDMA}}}^{{\rm{case2}}} \ge R_{{\rm{NOMA}}}^{{\rm{case2}}}$, we have $R_{{\rm{TDMA}}}^{{\rm{case2}}} = R_{{\rm{NOMA}}}^{{\rm{case2}}}$, which thus completes the proof.
\vspace{-8pt}
\section{Proof of Proposition 1}
The partial Lagrangian function of problem \eqref{SURA} can be directly written as
\begin{align}\label{partial lagrangian}
{\cal L}\left( \Xi  \right) = B{\tau _1}{\log _2}\left( {1 + \frac{{eh}}{{{\tau _1}{\sigma ^2}}}} \right) + \frac{{Tf}}{C} + \lambda \left( {\eta {\tau _0}{P_E}h - {e_1} - T{\gamma _c}f_1^3} \right){\rm{ + }}\mu \left( {T - {\tau _0} - {\tau _1}} \right),
\end{align}
where  $\Xi  = \left\{ {{\tau _0},{\tau _1},e,f,\lambda ,u} \right\}$, $\lambda  \ge 0 $ and $\mu \ge 0$ are the corresponding Lagrange multipliers. For a convex problem, the optimal solution can be obtained through analyzing the KKT conditions. Taking the partial derivative of ${\cal L}\left( \Xi  \right)$  with respect to ${{\tau _1}}$, ${{\tau _0}}$, $e$, and $f$, respectively, we have
\begin{align}\label{KKT}
&\frac{{\partial {\cal L}\left( \Xi  \right)}}{{\partial {\tau _1}}} = B{\log _2}\left( {1 + \frac{{eh}}{{{\tau _1}{\sigma ^2}}}} \right) - \frac{{Beh}}{{\left( {{\tau _1}{\sigma ^2} + eh} \right)\ln 2}} - u,\\
&\frac{{\partial {\cal L}\left( \Xi  \right)}}{{\partial {\tau _0}}} = \lambda \eta {P_E}h - u,\\
&\frac{{\partial {\cal L}\left( \Xi  \right)}}{{\partial e}} = \frac{{B{\tau _1}h}}{{\left( {{\tau _1}{\sigma ^2} + eh} \right)\ln 2}} - \lambda,\\
&\frac{{\partial {\cal L}\left( \Xi  \right)}}{{\partial f}}{\rm{ = }}\frac{T}{C} - 3\lambda T{\gamma _c}f.
\end{align}

Note that ${\tau _0} > 0$ always holds at the optimal solution. If UL offloading is activated at the optimal solution, i.e., ${\tau _1} > 0,e > 0$, we have $\frac{{\partial {\cal L}\left( \Xi  \right)}}{{\partial {\tau _0}}} = 0$, $\frac{{\partial {\cal L}\left( \Xi  \right)}}{{\partial {\tau _1}}} = 0$, $\frac{{\partial {\cal L}\left( \Xi  \right)}}{{\partial e}} = 0$, and $\frac{{\partial {\cal L}\left( \Xi  \right)}}{{\partial f}} = 0$. In this case, the optimal transmit power denoted by ${p^*}$  satisfies
\begin{align}\label{optimal_p}
{\log _2}\left( {1 + \frac{{{p^*}h}}{{{\sigma ^2}}}} \right) - \frac{{{p^*}h}}{{\left( {{\sigma ^2} + {p^*}h} \right)\ln 2}} - \sum\limits_{k = 1}^K {\frac{{\eta {P_E}{h^2}}}{{\left( {{\sigma ^2} + {p^*}h} \right)\ln 2}}}  = 0,
\end{align}
which yields~\eqref{C30dev}. Accordingly, the optimal local computing frequency denoted by ${f^*}$ is
\begin{align}\label{optimal_f}
{f^*} = \sqrt {\frac{1}{{3C{\lambda ^*}{\gamma _c}}}}  = \sqrt {\frac{{\left( {{\sigma ^2} + {p^*}h} \right)\ln 2}}{{3Ch{\gamma _c}B}}}.
\end{align}
Since the device depletes all of its harvested energy and ${\tau _0} + {\tau _1} = T$, we have
\begin{align}\label{optimal_tao1}
{\tau _1} = \frac{{\eta {P_E}h - {\gamma _c}{{\left( {\frac{{\left( {{\sigma ^2} + {p^*}h} \right)\ln 2}}{{3Ch{\gamma _c}B}}} \right)}^{\frac{3}{2}}}}}{{{p^*} + \eta {P_E}h}}T.
\end{align}
According to~\eqref{optimal_tao1}, the condition in~\eqref{C31dev} must hold if UL offloading is activated, i.e., ${\tau _1} > 0$. On the other hand, if the condition in \eqref{C31dev} holds, a non-trivial ${\tau _1} > 0$ can be obtained from \eqref{optimal_tao1}, which satisfies KKT conditions and it is hence optimal for problem \eqref{SURA} due to its convexity.
\end{appendices}

\bibliographystyle{IEEEtran}
\vspace{-8pt}
\bibliography{myref}


\end{document}